\documentclass[journal]{IEEEtran}

\usepackage{graphicx}				% Use pdf, png, jpg, or eps§ with pdflatex; use eps in DVI mode
\IEEEoverridecommandlockouts                              % This command is only
\usepackage{amsmath}
\usepackage{amssymb,amsthm}
\usepackage{algorithm}
\usepackage{color}

\newtheorem{Thm}{Theorem}
\newtheorem{Fact}{Fact}
\newtheorem{Lem}{Lemma}
\newtheorem{Cor}{Corollary}

\newtheorem{Claim}{Claim}

\DeclareMathOperator*{\argmin}{argmin}

\newcommand*{\herm}{^{\mkern-1.5mu\mathsf{H}}}
\newcommand*{\tran}{^{\mkern-1.5mu\mathsf{T}}}

\begin{document}

\title{\huge Dynamic Transmit Covariance Design in MIMO Fading Systems  With Unknown Channel Distributions and Inaccurate Channel State Information}

%\author{Hao Yu,~\IEEEmembership{Member,~IEEE,}
%        and~Michael~J.~Neely,~\IEEEmembership{Senior~Member,~IEEE}% <-this % stops a space
%\IEEEcompsocitemizethanks{\IEEEcompsocthanksitem Hao Yu and Michael J. Neely are with the Department of Electrical Engineering, University of Southern California, Los Angeles, CA 90089.}
%% note need leading \protect in front of \\ to get a newline within \thanks as
%% \\ is fragile and will error, could use \hfil\break instead.
%\thanks{This work was presented in part at IEEE International Conference on Computer Communications (INFOCOM), San Francisco, CA, USA, April, 2016 \cite{YuNeelyINFOCOM16}. This work is supported in part by the NSF grant CCF-0747525.}}

\author{Hao~Yu,~\IEEEmembership{Student Member,~IEEE,} 
	Michael~J.~Neely,~\IEEEmembership{Senior Member,~IEEE,}% <-this % stops a space
\thanks{Hao Yu  and Michael J. Neely are with the Department of Electrical Engineering, University of Southern California, Los Angeles, CA, USA.}% <-this % stops a space
\thanks{This work was presented in part at IEEE International Conference on Computer Communications (INFOCOM), San Francisco, CA, USA, April, 2016 \cite{YuNeely16INFOCOM}. This work is supported in part by the NSF grant CCF-0747525.}}

%\markboth{IEEE Transactions on Wireless Communications, to appear}%

\maketitle

\begin{abstract}
This paper considers dynamic  transmit covariance design in point-to-point MIMO fading systems with unknown channel state distributions and inaccurate channel state information  subject to both long term and short term power constraints. First, the case of instantaneous but possibly inaccurate channel state information at the transmitter (CSIT) is treated.  By extending the drift-plus-penalty technique, a dynamic transmit covariance policy is developed and is shown to approach optimality with an $O(\delta)$ gap, where $\delta$ is the inaccuracy measure of CSIT, regardless of the channel state distribution and without requiring knowledge of this distribution.  Next, the case of delayed and inaccurate channel state information is considered.  The optimal transmit covariance solution that maximizes the ergodic capacity is fundamentally different in this case, and a different online algorithm based on convex projections is developed.  The proposed algorithm for this delayed-CSIT case also has an $O(\delta)$ optimality gap, where $\delta$ is again the  inaccuracy measure of CSIT. 
\end{abstract}

\section{Introduction}\label{sec:introduction}

\IEEEPARstart{D}{uring} the past decade, the multiple-input multiple-output (MIMO) technique has been recognized as one of the most important techniques for increasing the capabilities of wireless communication systems.  In the wireless fading channel, where the channel changes over time, the problem of transmit covariance design  is to determine the transmit covariance of the transmitter to maximize the capacity subject to both long term and short term power constraints. It is often  reasonable to assume that instantaneous channel state information (CSI) is available at the receiver through training.  Most works on transmit covariance design in MIMO fading systems also assume that statistical information about the channel state, referred to as channel distribution information (CDI), is available at the transmitter. Under the assumption of perfect channel state information at the receiver (CSIR) and perfect channel distribution information at the transmitter (CDIT), prior work on transmit covariance design in point-to-point MIMO fading systems can be grouped into two categories: 
\begin{itemize}
\item Instantaneous channel state information at the transmitter: In the ideal case of perfect\footnote{In this paper, CSIT is said to be ``perfect" if it is both instantaneous (i.e., has no delay) and accurate.} CSIT, optimal transmit covariance design for MIMO links with both long term and short term power constraints is a water-filling solution \cite{Telatar99MIMOCapacity}. Computation of water-levels involves a one-dimensional integral equation for  fading channels with independent and identically distributed (i.i.d.) Rayleigh entries or a multi-dimensional integral equation for general fading channels \cite{Jayaweera03IT}.  The involved multi-dimensional integration equation is in general intractable and can only be approximately solved with numerical algorithms with huge complexity. MIMO fading systems with dynamic CSIT is considered in \cite{Vu07JSAC}.

\item No CSIT:  If CSIT is unavailable, the optimal transmit covariance design is in general still open. If the channel matrix has i.i.d. Rayleigh entries, then the optimal transmit covariance  is known to be the identity transmit covariance scaled to satisfy the power constraint \cite{Telatar99MIMOCapacity}. The optimal transmit covariance in MIMO fading channels with correlated Rayleigh entries is obtained in \cite{Jafar01ICC, Jorswieck04TWC}. The transmit covariance design in MIMO fading channels is further considered in \cite{Veeravalli05IT} under a more general channel correlation model.
\end{itemize}

These prior works rely on accurate CDIT and/or on restrictive channel distribution assumptions. 
It can be difficult to accurately estimate the CDI,  especially when there are complicated correlations between entries in the channel matrix. Solutions that base decisions on CDIT can be suboptimal due to mismatches. Work \cite{Polomar03IT} considers MIMO fading channels without CDIT and aims to find the transmit covariance to maximize the worst channel capacity using a game theoretical approach rather than solve the original ergodic capacity maximization problem. In contrast, the current paper proposes algorithms that do not require prior knowledge of the channel distribution,  yet perform arbitrarily close to the optimal value of the ergodic capacity maximization that can be achieved by having CDI knowledge.  

In time-division duplex (TDD) systems with symmetric wireless channels, the CSI can be measured directly at the transmitter using the unlink channel. However, in frequency-division duplex  (FDD) scenarios and other scenarios without channel symmetry, the CSI must be measured at the receiver, quantized, and reported back to the transmitter with a time delay \cite{book_FundamentalWireless}.  

Depending on the measurement delay in TDD systems or the overall channel acquisition delay in FDD systems, the CSIT can be instantaneous or delayed.  In general, the CSIT can also be inaccurate due to the measurement, quantization or feedback error.   This paper first considers the instantaneous (but possibly inaccurate) CSIT case and develops an algorithm that does not require CDIT.  This algorithm can achieve a utility within  $O(\delta)$ of the best utility that can be achieved with CDIT and perfect CSIT, where $\delta$ is the inaccuracy measure of CSIT. This further implies that accurate instantaneous CSIT (with $\delta= 0$) is almost as good as having both CDIT and accurate instantaneous CSIT.

Next, the case of delayed (but possibly inaccurate) CSIT is considered and a fundamentally different algorithm is developed for that case.  The latter algorithm again does not use CDIT, but 
achieves a utility within $O(\delta)$ of the best utility that can be achieved even with CDIT, where  $\delta$ is the inaccuracy measure of CSIT. This further implies that delayed but accurate CSIT (with $\delta= 0$) is almost as good as having CDIT.

\subsection{Related work and our contributions}

In the instantaneous (and possibly inaccurate) CSIT case, the proposed dynamic transmit covariance design extends the general drift-plus-penalty algorithm for stochastic network optimization  \cite{PhD_Thesis_Neely,book_Neely10} to deal with inaccurate observations of system states. In this MIMO context, the current paper shows the algorithm provides strong sample path convergence time guarantees.  The dynamic of the drift-plus-penalty algorithm is similar to that of the stochastic dual subgradient algorithm, although the optimality analysis and performance bounds are different.   The stochastic dual subgradient algorithm has been applied to optimization in wireless fading channels without CDI, e.g., downlink power scheduling in single antenna cellular systems \cite{Shroff06TWC},  power allocation in single antenna broadcast OFDM channels \cite{Rebeiro10TSP},  scheduling and resource allocation in random access channels \cite{HuRibeiro11TWC}, transmit covariance design in multi-carrier MIMO networks \cite{Liu09Mobihoc}.  

In the delayed (and possibly inaccurate) CSIT case, the situation is similar to the scenario of online convex optimization \cite{Zinkevich03ICML} except that we are unable to observe true history reward functions due to channel error. The proposed dynamic power allocation policy can be viewed as an online algorithm with inaccurate history information. The current paper analyzes the performance loss due to CSIT inaccuracy and provides strong sample path convergence time guarantees of this algorithm. The analysis in this MIMO context can be extended to more general online convex optimization with inaccurate history information.  Online optimization has been applied in power allocation in wireless fading channels without CDIT and with delayed and accurate CSIT, e.g., suboptimal online power allocation in single antenna single user channels \cite{Buchbinder09INFOCOM}, suboptimal online power allocation in single antenna multiple user channels \cite{Buchbinder10INFOCOM}.  Online  transmit covariance design in MIMO systems with inaccurate CSIT is also considered in recent works  \cite{Stiakogiannakis15arXiv, Mertikopoulos16TSP, Mertikopoulos16JSAC}. The online algorithms in \cite{Stiakogiannakis15arXiv, Mertikopoulos16TSP, Mertikopoulos16JSAC} follow either a matrix exponential learning scheme or an online projected gradient scheme. However, all of these works assume that the imperfect CSIT is unbiased, i.e., expected CSIT error conditional on observed previous CSIT is zero. This assumption of imperfect CSIT is suitable when modeling the CSIT measurement error or feedback error but cannot capture the CSI quantization error. In contrast, the current paper only requires that CSIT error is bounded. 

\section{Signal model and problem formulations}

\subsection{Signal model} \label{sec:signal-model}

Consider a point-to-point MIMO block fading channel with $N_{T}$ transmit antennas  and $N_{R}$ receive antennas. In a block fading channel model, the channel matrix remains constant at each block and changes from block to block in an independent and identically distributed (i.i.d.) manner. Throughout this paper, each block is also called a slot and is assigned an index $t \in \{0, 1, 2, \ldots\}$.  At each slot $t$, the received signal \cite{Telatar99MIMOCapacity} is described by
\begin{align*}
\mathbf{y}(t) = \mathbf{H}(t) \mathbf{x}(t) + \mathbf{z}(t)
\end{align*}
where $t \in \{0, 1, 2, \ldots\}$ is the time index, $\mathbf{z}(t) \in \mathbb{C}^{N_{R}}$ is the  additive noise vector, $\mathbf{x}(t) \in \mathbb{C}^{N_{T}}$ is the transmitted signal vector, $\mathbf{H}(t) \in \mathbb{C}^{N_{R}\times N_{T}}$ is the channel matrix, and $\mathbf{y}(t) \in \mathbb{C}^{N_{R}}$ is the received signal vector.  Assume that noise vectors $\mathbf{z}(t)$ are i.i.d. normalized circularly symmetric complex Gaussian random vectors with $\mathbb{E}[\mathbf{z}(t) \mathbf{z}\herm(t)] = \mathbf{I}_{N_{R}}$, where $\mathbf{I}_{N_{R}}$ denotes an $N_{R}\times N_{R}$ identity matrix.\footnote{If the size of the identity matrix is clear, we often simply write $\mathbf{I}$.} Note that channel matrices $\mathbf{H}(t)$ are i.i.d. across slot $t$ and have a fixed but arbitrary probability distribution, possibly one with 
correlations between entries of the matrix. Assume there is a constant $B>0$ such that $\Vert \mathbf{H}\Vert_{F} \leq B$ with probability one, where $\Vert \cdot\Vert_{F}$ denotes the Frobenius norm.\footnote{A bounded Frobenius norm always holds in the physical world 
because the  channel attenuates the signal. Particular models such as 
Rayleigh and Rician fading violate this assumption in order to have simpler distribution functions \cite{Bai03ComSurvey}.} Recall that the Frobenius norm of a complex $m\times n$ matrix $\mathbf{A} = (a_{ij})$ is 
\begin{equation} \label{eq:frobenius-def} 
 \Vert \mathbf{A}\Vert_{F}  = \mbox{$\sqrt{\sum_{i=1}^m \sum_{j=1}^n |a_{ij}|^2}$} = \sqrt{\text{tr}(\mathbf{A\herm}\mathbf{A})} 
\end{equation} 
where $\mathbf{A}\herm$ is the Hermitian transpose of $\mathbf{A}$ and $\text{tr}(\cdot)$ is the trace operator. 

Assume that the receiver can track $\mathbf{H}(t)$ exactly at each slot $t$ and hence has perfect CSIR.  In practice, CSIR is obtained by sending designed training sequences, also known as pilot sequences, which are commonly known to both the transmitter and the receiver, such that the channel matrix $\mathbf{H}(t)$ can be estimated at the receiver \cite{book_FundamentalWireless}. CSIT is obtained in different ways in different wireless systems. In TDD systems, the transmitter exploits channel reciprocity and use the measured uplink channel as approximated CSIT. In FDD systems, the receiver creates a quantized version of CSI, which is a function of $\mathbf{H}(t)$, and reports back to the transmitter after a certain amount of delay. In general, there are two possibilities of CSIT availabilities:
\begin{itemize}
\item {\bf Instantaneous CSIT Case}:  In TDD systems or FDD systems where the measurement, quantization and feedback delays are negligible with respect to the channel coherence time, an approximate version $\widetilde{\mathbf{H}}(t)$ for the true channel $\mathbf{H}(t)$ is known at the transmitter at each time slot $t$.
\item {\bf Delayed CSIT Case}: In FDD systems with a large CSIT acquisition delay, the transmitter only knows $\widetilde{\mathbf{H}}(t-1)$, which is an approximate version of channel $\mathbf{H}(t-1)$, and does not know $\mathbf{H}(t)$ at each time slot $t$.\footnote{In general, the dynamic transmit covariance design developed in this paper can be extended to deal with arbitrary CSIT acquisition delay as discussed in Section \ref{sec:delay-csit-extension}. For the simplicity of presentations, we assume the CSIT acquisition delay is always one slot in this paper.}  
\end{itemize}

In both cases, we assume the CSIT inaccuracy is bounded, i.e., there exists $\delta > 0$ such that $\Vert \widetilde{\mathbf{H}}(t) - \mathbf{H}(t)\Vert_{F} \leq \delta$ for all $t$. 

\subsection{Problem Formulation}
At each slot $t$, if the channel matrix is  $\mathbf{H}(t)$ and the transmit covariance is  $\mathbf{Q}(t)$, then the  MIMO capacity is given by \cite{Telatar99MIMOCapacity}:  
\[ \log\det(\mathbf{I} + \mathbf{H}(t)\mathbf{Q}(t) \mathbf{H}\herm(t)) \]
where $\det(\cdot)$ denotes the determinant operator of matrices. The (long term) average capacity\footnote{The expression $\mathbb{E}_{\mathbf{H}} \big[ \log\det(\mathbf{I} + \mathbf{H}\mathbf{Q} \mathbf{H}\herm) \big]$ is also known as the ergodic capacity. In fast fading channels where the channel coherence time is smaller than the codeword length, ergodic capacity can be attained if each codeword spans across sufficiently many channel blocks. In slow fading channels where the channel coherence time is larger than the codeword length, ergodic capacity can be attained by adapting both transmit covariances and data rates to the CSIT of each channel block (see \cite{book_Lau06} for related discussions). In slow fading channels, the ergodic capacity is essentially the long term average capacity since it is asymptotically equal to the average capacity of each channel block (by the law of large numbers).  Note that another concept ``outage capacity" is sometimes considered for slow fading channels when there is no rate adaptation and the data rate is constant regardless of channel realizations (In this case, the data rate can be larger than the block capacity for poor channel realizations such that ``outage"  occurs). In this paper, we have both transmit covariance design and rate adaptation; and hence consider ``ergodic capacity".} of the MIMO block fading channel \cite{book_WirelessCom} is given by 
$$ \mathbb{E}_{\mathbf{H}} \big[ \log\det(\mathbf{I} + \mathbf{H}\mathbf{Q} \mathbf{H}\herm) \big]$$
where $\mathbf{Q}$ can adapt to $\mathbf{H}$ when CSIT is available and is a constant matrix when CSIT is unavailable. Consider two types of power constraints at the transmitter: A long term average power constraint $\mathbb{E}_{\mathbf{H}}[ \text{tr}(\mathbf{Q})] \leq \bar{P}$ and a short term power constraint $\text{tr}(\mathbf{Q})\leq P$ enforced at each slot. The long term constraint arises from battery or energy limitations while the short term constraint is often due to hardware or regulation limitations. 

If CSIT is available, the problem is to choose $\mathbf{Q}$ as a (possibly random) function of the observed $\mathbf{H}$ to maximize the (long term) average capacity subject to both power constraints: 
\begin{align}
\max_{\mathbf{Q}(\mathbf{H})} \quad &  \mathbb{E}_{\mathbf{H}} \big[ \log\det(\mathbf{I} + \mathbf{H}\mathbf{Q}(\mathbf{H}) \mathbf{H}\herm) \big]  \label{eq:stochastic-with-csit-obj}\\
\text{s.t.} \quad  & \mathbb{E}_{\mathbf{H}}[\text{tr}(\mathbf{Q}(\mathbf{H}))] \leq \bar{P}, \label{eq:stochastic-with-csit-total-power-con}\\
			 & \mathbf{Q}(\mathbf{H}) \in \mathcal{Q}, \forall \mathbf{H}, \label{eq:stochastic-with-csit-set-con}
\end{align}
where $\mathcal{Q}$ is a set that enforces the short term power constraint: 
\begin{align} \label{eq:Q-set} 
\mathcal{Q} = \big\{\mathbf{Q}\in \mathbb{S}^{N_{T}}_{+}: \text{tr}(\mathbf{Q}) \leq P \big\}
\end{align} 
where  $\mathbb{S}^{N_{T}}_{+}$ denotes the $N_{T} \times N_{T}$ positive semidefinite matrix space.  To avoid trivialities, we assume that $P \geq \bar{P}$. In \eqref{eq:stochastic-with-csit-obj}-\eqref{eq:stochastic-with-csit-set-con}, we use notation $\mathbf{Q}(\mathbf{H})$ to emphasize that $\mathbf{Q}$ can depend on $\mathbf{H}$, i.e., adapt to channel realizations.  Under the long term power constraint, the optimal power allocation should be opportunistic, i.e., use more power over good channel realizations and less power over poor channel realizations. It is known that opportunistic power allocation provides a significant capacity gain in low SNR regimes and a marginal gain in high SNR regimes compared with fixed power allocation \cite{book_FeedbackWirelessComm}.

Without CSIT, the optimal  transmit covariance design problem is different, given as follows. 
\begin{align}
\max_{\mathbf{Q}} \quad &  \mathbb{E}_{\mathbf{H}} \big[ \log\det(\mathbf{I} + \mathbf{H}\mathbf{Q} \mathbf{H}\herm) \big]  \label{eq:stochastic-no-csit-obj}\\
\text{s.t.} \quad  & \mathbb{E}_{\mathbf{H}}[\text{tr}(\mathbf{Q})] \leq \bar{P},  \label{eq:stochastic-no-csit-total-power-con}\\
			 & \mathbf{Q} \in \mathcal{Q}, \label{eq:stochastic-no-csit-set-con}
\end{align}
where set $\mathcal{Q}$ is defined in \eqref{eq:Q-set}. Again assume $P\geq \bar{P}$. Since the instantaneous CSIT is unavailable, the transmit covariance cannot adapt  to $\mathbf{H}$.  By the convexity of this problem and Jensen's inequality, a randomized $\mathbf{Q}$ is useless. It suffices to consider a constant $\mathbf{Q}$.  Since $P\geq \bar{P}$, this implies the problem is equivalent to a problem that removes the constraint
\eqref{eq:stochastic-no-csit-total-power-con} and that changes the constraint 
\eqref{eq:stochastic-no-csit-set-con} to: 
\[ \mathbf{Q} \in \widetilde{\mathcal{Q}}= \{\mathbf{Q} \in \mathbb{S}_{+}^{N_T} : \text{tr}(\mathbf{Q}) \leq \bar{P}\} \]

The problems \eqref{eq:stochastic-with-csit-obj}-\eqref{eq:stochastic-with-csit-set-con}  and \eqref{eq:stochastic-no-csit-obj}-\eqref{eq:stochastic-no-csit-set-con} are fundamentally different and have different optimal objective function values. Most existing works \cite{Jayaweera03IT,Jafar01ICC, Jorswieck04TWC,Veeravalli05IT} on MIMO fading channels can be interpreted as solutions to either of the above two stochastic optimization under specific channel distributions. Moreover, those works require perfect channel distribution information (CDI). In this paper, the above two stochastic optimization problems are solved via dynamic algorithms that works for arbitrary channel distributions and does not require any CDI. The algorithms are different for the two cases, and use different techniques.

\section{Instantaneous CSIT case}
Consider the case of instantaneous but inaccurate CSIT where at each slot $t \in \{0, 1, 2,\ldots\}$, channel $\mathbf{H}(t)$ is unknown and only an approximate version $\widetilde{\mathbf{H}}(t)$ is known. In this case, the problem \eqref{eq:stochastic-with-csit-obj}-\eqref{eq:stochastic-with-csit-set-con} can be interpreted as a stochastic optimization problem where channel $\mathbf{H}(t)$ is the instantaneous system state and transmit covariance $\mathbf{Q}(t)$ is the control action at each slot $t$. This is similar to the scenario of stochastic optimization with i.i.d. time-varying  system states, where the decision maker chooses an action based on the observed instantaneous system state at each slot such that time average expected utility is maximized and the time average expected constraints are guaranteed. The {\it drift-plus-penalty (DPP)} technique from \cite{book_Neely10} is a mature framework to solve stochastic optimization without distribution information of system states.

This is different from the conventional stochastic optimization considered by the DPP technique because at each slot $t$, the true ``system state" $\mathbf{H}(t)$ is unavailable and only an approximate version $\widetilde{\mathbf{H}}(t)$ is known. Nevertheless, a modified version of the standard DPP algorithm is developed in Algorithm \ref{alg:with-csit}.

\begin{algorithm}
\caption{Dynamic Transmit Covariance Design with instantaneous CSIT}
\label{alg:with-csit}
Let $V>0$ be a constant parameter and $Z(0)=0$. At each time $t\in\{0,1,2,\ldots\}$, observe $\widetilde{\mathbf{H}}(t)$ and $Z(t)$.  Then do the following:
\begin{itemize}
\item Choose transmit covariance $\mathbf{Q}(t) \in \mathcal{Q}$ to solve : 
\[ \max_{\mathbf{Q}\in  \mathcal{Q}}\{V\log\det(\mathbf{I} + \widetilde{\mathbf{H}}(t) \mathbf{Q} \widetilde{\mathbf{H}}\herm(t) - Z(t) \text{tr}(\mathbf{Q})\}.  \]
\item Update $Z(t+1) = \max[0,Z(t) + \text{tr}(\mathbf{Q}(t)) - \bar{P}]$.
\end{itemize}
\end{algorithm}

In Algorithm \ref{alg:with-csit}, a \emph{virtual queue} $Z(t)$ with $Z(0)=0$ and with update $Z(t+1) = \max[0,Z(t) + \text{tr}(\mathbf{Q}(t)) - \bar{P}]$ is introduced to enforce the average power constraint \eqref{eq:stochastic-with-csit-total-power-con} and can be viewed as the ``queue backlog" of long term power constraint violations since it increases at slot $t$ if the power consumption at slot $t$ is larger than $\bar{P}$ and decreases otherwise. The next Lemma relates $Z(t)$ and the average power consumption. 

\begin{Lem} \label{lm:with-csit-constaint-and-queue-relation}
Under Algorithm \ref{alg:with-csit}, we have $$\frac{1}{t} \sum_{\tau=0}^{t-1} \text{tr}(\mathbf{Q}(\tau)) \leq \bar{P} + \frac{Z(t)}{t}, \quad \forall t>0.$$
\end{Lem}
\begin{IEEEproof}
Fix $t>0$. For all slots $\tau\in\{0,1,\ldots,t-1\}$, the update for $Z(\tau)$ satisfies $Z(\tau+1) = \max[0,Z(\tau) + \text{tr}(\mathbf{Q}(\tau)) - \bar{P}] \geq Z(\tau) +\text{tr}(\mathbf{Q}(\tau)) - \bar{P}$. Rearranging terms gives: $ \text{tr}(\mathbf{Q}(\tau))  \leq  \bar{P} + 
Z(\tau+1) -Z(\tau)$.  Summing over $\tau \in \{0, \ldots, t-1\}$ and dividing by factor $t$ gives: 
\begin{align*}
\frac{1}{t} \sum_{\tau=0}^{t-1} \text{tr}(\mathbf{Q}(\tau)) &\leq  \bar{P} + \frac{Z(t)-Z(0)}{t} \overset{(a)}{=} \bar{P} + \frac{Z(t)}{t}
\end{align*}
where (a) follows from $Z(t)=0$.
\end{IEEEproof}

For each slot $t \in \{0, 1, 2, \ldots\}$ define the \emph{reward} $R(t)$: 
\begin{equation} \label{eq:Rt} 
R(t) = \log\det(\mathbf{I} + \mathbf{H}(t)\mathbf{Q}(t) \mathbf{H}\herm(t)). 
\end{equation} 
Define $R^{\text{opt}}$ as the optimal average utility in \eqref{eq:stochastic-with-csit-obj}. The value $R^{\text{opt}}$ depends on the (unknown) distribution for $\mathbf{H}(t)$. Fix $\epsilon>0$ and define $V=\max\{\bar{P}^2, (P-\bar{P})^2\} /(2\epsilon)$.  If $\widetilde{\mathbf{H}}(t) = \mathbf{H}(t), \forall t$, regardless of the distribution of $\mathbf{H}(t)$, the standard DPP technique \cite{book_Neely10} ensures: 
\begin{align}
 &\frac{1}{t} \sum_{\tau=0}^{t-1} \mathbb{E}[R(\tau)] \geq R^{\text{opt}} - \epsilon, \qquad \forall t>0  \label{eq:book1}  \\
&\lim_{t\rightarrow\infty} \frac{1}{t}\sum_{\tau=0}^{t-1} \mathbb{E}[\text{tr}(\mathbf{Q}(\tau))] \leq \bar{P} \label{eq:book2} 
\end{align} 
This holds for arbitrarily small values of $\epsilon>0$, and so the algorithm comes arbitrarily close to optimality.  However, the above is true only if  $\widetilde{\mathbf{H}}(t) = \mathbf{H}(t), \forall t$.

The development and analysis of Algorithm \ref{alg:with-csit} extends the DPP technique in two aspects:
\begin{itemize}
\item At each slot $t$, the standard drift-plus-penalty technique requires accurate ``system state"  $\mathbf{H}(t)$ and cannot deal with inaccurate ``system state" $\widetilde{\mathbf{H}}(t)$.  In contrast, Algorithm \ref{alg:with-csit} works with $\widetilde{\mathbf{H}}(t)$. The next subsections show that the performance of Algorithm \ref{alg:with-csit} degrades linearly with respect to CSIT  inaccuracy measure  $\delta$. If $\delta =0$, then \eqref{eq:book1} is recovered.
\item Inequality \eqref{eq:book2} only treats
infinite horizon time average expected power.  The next subsections show that Algorithm \ref{alg:with-csit} can guarantee $\frac{1}{t}\sum_{\tau=0}^{t-1} \text{tr}(\mathbf{Q}(\tau)) \leq \bar{P} + \frac{(B+\delta)^2\max\{\bar{P}^2, (P-\bar{P})^2\} + 2\epsilon (P-\bar{P})}{2\epsilon t}$ for all $t>0$. This sample path guarantee on average power consumption is much stronger than \eqref{eq:book2}.  In fact, \eqref{eq:book2} is recovered by taking expectation and taking limit $t\rightarrow \infty$.
\end{itemize}
\subsection{Transmit covariance updates in Algorithm \ref{alg:with-csit}}

This subsection shows the  $\mathbf{Q}(t)$ selection in Algorithm \ref{alg:with-csit} 
has an (almost) closed-form solution.  
The convex program involved in the transmit covariance update of Algorithm \ref{alg:with-csit} is in the form
\begin{align}
\max_{\mathbf{Q}} \quad & \log\det(\mathbf{I} + \mathbf{H} \mathbf{Q} \mathbf{H}\herm) - \frac{Z}{V}\text{tr}(\mathbf{Q})  \label{eq:with-csit-primal-opt-obj}\\
\text{s.t.} \quad  & \text{tr}(\mathbf{Q}) \leq P  \label{eq:with-csit-primal-opt-trace}\\
			 & \mathbf{Q} \in \mathbb{S}^{N_{T}}_+ \label{eq:with-csit-primal-opt-sdp}
\end{align}
This convex program is similar to the conventional problem of transmit covariance design with a deterministic channel $\mathbf{H}$, except that objective \eqref{eq:with-csit-primal-opt-obj} has an additional penalty term $-(Z/V) \text{tr}(\mathbf{Q})$. It is well known that, without this penalty term,  the solution is to diagonalize the channel matrix and allocate power over eigen-modes according to a water-filling technique \cite{Telatar99MIMOCapacity}. The next lemma summarizes that the optimal solution to problem \eqref{eq:with-csit-primal-opt-obj}-\eqref{eq:with-csit-primal-opt-sdp} has a similar structure.

\begin{Lem}\label{lm:with-csit-transmit-covaraince-update}
Consider the SVD $\mathbf{H}\herm\mathbf{H} = \mathbf{U}\herm\mathbf{\Sigma}\mathbf{U}$, where $\mathbf{U}$ is a unitary matrix and $\boldsymbol{\Sigma}$ is a diagonal matrix with non-negative entries $\sigma_{1}, \ldots, \sigma_{N_{T}}$.  
Then the optimal solution to \eqref{eq:with-csit-primal-opt-obj}-\eqref{eq:with-csit-primal-opt-sdp}
is given by $\mathbf{Q}^* = \mathbf{U}\herm\mathbf{\Theta}^*\mathbf{U}$, where $\mathbf{\Theta}^*$ is a diagonal matrix with entries $\theta_1^*, \ldots, \theta_{N_T}^*$ given by: 
\[ \theta_i^* = \max\big[0, \frac{1}{\mu^* + Z/V} - \frac{1}{\sigma_i}\big], \quad \forall i \in \{1, \ldots, N_T\}, \]
where $\mu^{\ast}$ is chosen such that $\sum_{i=1}^{N_T} \theta_i^\ast \leq  P$, $\mu^{\ast} \geq 0$ and $\mu^{\ast} \big[\sum_{i=1}^{N_{T}} \theta_i^\ast -  P\big] = 0$.  The exact $\mu^*$ can be determined using Algorithm \ref{alg:with-csit-primal-opt-exact-alg} with complexity $O(N_T\log N_T)$.
\end{Lem}

\begin{IEEEproof}
The proof is a simple extension of the classical proof for the optimal transmit covariance in deterministic MIMO channels, e.g. Section 3.2 in \cite{Telatar99MIMOCapacity}, to deal with the additional penalty term $-(Z/V) \text{tr}(\mathbf{Q})$. See Appendix \ref{app:pf-lm-with-csit-transmit-covaraince-update} for a complete proof.
\end{IEEEproof}

\begin{algorithm} 
\caption{ Algorithm  to solve problem \eqref{eq:with-csit-primal-opt-obj}-\eqref{eq:with-csit-primal-opt-sdp}} 
\label{alg:with-csit-primal-opt-exact-alg}
\begin{enumerate}
\item Check if $\sum_{i=1}^{N_{T}} \max\{0,\frac{1}{Z/V} - \frac{1}{\sigma_{i}}\} \leq P$ holds. If yes, let $\mu^\ast = 0$ and $\theta_{i}^{\ast} =\max\{0,\frac{1}{Z/V} - \frac{1}{\sigma_{i}}\}, \forall i\in\{1,2,\ldots,N_{T}\}$ and terminate the algorithm; else, continue to the next step.
\item Sort all $\sigma_{i}, \in\{1,2,\ldots,N_{T}\}$ in a decreasing order $\pi$ such that $\sigma_{\pi(1)} \geq \sigma_{\pi(2)} \geq \cdots \geq \sigma_{\pi(N_{T})}$. Define $S_{0}=0$.
\item For $i=1$ to $N_{T}$
\begin{itemize}
\item  Let $S_{i} = S_{i-1} + \frac{1}{\sigma_{\pi(i)}}$. Let $\mu^{\ast} = \frac{i}{S_{i} +P} - (Z/V)$.
\item  If $\mu^{\ast}\geq 0$, $\frac{1}{\mu^{\ast} + Z/V} - \frac{1}{\sigma_{\pi(i)}}>0$ and $\frac{1}{\mu^{\ast} + Z/V} - \frac{1}{\sigma_{\pi(i+1)}} \leq 0$, then terminate the loop; else, continue to the next iteration in the loop. 
\end{itemize}
\item Let $\theta_{i}^{\ast} =\max\big[0,\frac{1}{\mu^\ast+ Z/V} - \frac{1}{\sigma_{i}}\big] , \forall i\in\{1,2,\ldots,N_{T}\}$ and terminate the algorithm.
\end{enumerate}
\end{algorithm}
The complexity of Algorithm \ref{alg:with-csit-primal-opt-exact-alg} is dominated by the sorting of all $\sigma_{i}$ in step (2). Recall that the water-filling solution of power allocation in multiple parallel channels can also be found by an exact algorithm (see Section 6 in \cite{Palomar03JSAC}), which is similar to Algorithm \ref{alg:with-csit-primal-opt-exact-alg}. The main difference is that Algorithm \ref{alg:with-csit-primal-opt-exact-alg} has a first step to verify if $\mu^{\ast} = 0$. This is because unlike the power allocation in multiple parallel channels, where the optimal solution always uses full power, the optimal solution to problem \eqref{eq:with-csit-primal-opt-obj}-\eqref{eq:with-csit-primal-opt-sdp} may not use full power for large $Z$ due to the penalty term $-(Z/V) \text{tr}(\mathbf{Q})$ in objective \eqref{eq:with-csit-primal-opt-obj}.

\subsection{Performance of Algorithm \ref{alg:with-csit}} 
Define a Lyapunov function $L(t) = \frac{1}{2} Z^2(t)$ and its corresponding Lyapunov drift $\Delta(t) =  L(t+1) - L(t)$. The expression $-\Delta(t)+ VR(t)$ is called the DPP expression.   The analysis of the standard drift-plus-penalty (DPP) algorithm with accurate ``system states" relies on an upper bound of the DPP expression in terms of $R^{\text{opt}}$ \cite{book_Neely10}.  The performance analysis of Algorithm \ref{alg:with-csit}, which can be viewed as a DPP algorithm based on inaccurate ``system states", requires a new bound of the DPP expression in Lemma \ref{lm:with-csit-dpp-bound} and a new deterministic bound of virtual queue $Z(t)$ in Lemma  \ref{lm:with-csit-deterministic-queue-bound}.   
\begin{Lem}\label{lm:with-csit-dpp-bound}
Under Algorithm \ref{alg:with-csit}, we have
\begin{align*}
&-\mathbb{E}[\Delta(t)] +V\mathbb{E}[R(t)] \\
\geq& VR^{\text{opt}} - \frac{1}{2}\max\{\bar{P}^2, (P-\bar{P})^2\}  - 2VP\sqrt{N_T} (2B+\delta)\delta,
\end{align*}
where $B, \delta,N_T, P$ and $\bar{P}$ are defined in Section \ref{sec:signal-model}; and $R^{\text{opt}}$ is the optimal average utility in problem \eqref{eq:stochastic-with-csit-obj}-\eqref{eq:stochastic-with-csit-set-con}.
\end{Lem}
\begin{IEEEproof}
See Appendix \ref{app:pf-lm-with-csit-dpp-bound}.
\end{IEEEproof}

\begin{Lem} \label{lm:with-csit-deterministic-queue-bound}
Under Algorithm \ref{alg:with-csit}, we have  $Z(t) \leq V(B+\delta)^{2} + (P-\bar{P}), \forall t>0,$ where $B, \delta, P$ and $\bar{P}$ are defined in Section \ref{sec:signal-model}. 
\end{Lem}
\begin{IEEEproof}
We first show that if $Z(t) \geq V(B+\delta)^2$, then Algorithm \ref{alg:with-csit} chooses $\mathbf{Q}(t) = \mathbf{0}$.  Consider $Z(t)\geq V(B+\delta)^2$. Let SVD $\widetilde{\mathbf{H}}\herm(t)\widetilde{\mathbf{H}}(t) = \mathbf{U}\herm\boldsymbol{\Sigma}\mathbf{U}$, where diagonal matrix $\boldsymbol{\Sigma}$ has non-negative diagonal entries $\sigma_{1}, \ldots, \sigma_{N_{T}}$. Note that $\forall i\in\{1,2,\ldots, N_{T}\}$, $\sigma_i \overset{(a)}{\leq} \text{tr}(\widetilde{\mathbf{H}}\herm(t)\widetilde{\mathbf{H}}(t)) \overset{(b)}{=} \Vert \widetilde{\mathbf{H}}(t)\Vert_F^{2} \leq (\Vert \mathbf{H}(t)\Vert_F + \Vert \widetilde{\mathbf{H}}(t) - \mathbf{H}(t)\Vert_F )^2 \leq (B+\delta)^{2}$ where (a) follows from $\text{tr}(\widetilde{\mathbf{H}}\herm(t)\widetilde{\mathbf{H}}(t)) = \sum_{i=1}^{N_{T}} \sigma_{i}$; and (b) follows from the definition of Frobenius norm. By Lemma \ref{lm:with-csit-transmit-covaraince-update}, $\mathbf{Q}(t) = \mathbf{U}\herm\boldsymbol{\Theta}^{\ast} \mathbf{U}$, where $\boldsymbol{\Theta}^{\ast}$ is diagonal with entries $\theta_{1}^{\ast}, \ldots, \theta_{N_{T}}^{\ast}$ given by $\theta_i^{\ast} = \max\big[0,\frac{1}{\mu^{\ast} + Z(t)/V} - \frac{1}{\sigma_{i}}\big]$, where $\mu^{\ast}\geq 0$.  Since $\sigma_{i} \leq (B+\delta)^{2}, \forall i\in\{1,2,\ldots, N_{T}\}$, it follows that if $Z(t) \geq V(B+\delta)^{2}$, then $\frac{1}{\mu + Z(t)/V} - \frac{1}{\sigma_{i}}\leq 0$ for all $\mu\geq 0$ and hence $\theta_{i}^{\ast} = 0, \forall i\in\{1,2,\ldots,N_{T}\}$. This implies Algorithm \ref{alg:with-csit} chooses $\mathbf{Q}(t)= 0$ by Lemma \ref{lm:with-csit-transmit-covaraince-update}, which further implies that  $Z(t+1)\leq Z(t)$ by the update equation of $Z(t+1)$. 

On the other hand,  if $Z(t) \leq V (B+\delta)^{2}$, then $Z(t+1)$ is at most $V(B+\delta)^{2} + (P-\bar{P})$ by the update equation of $Z(t+1)$ and the short term power constraint $\text{tr}(\mathbf{Q}(t))\leq P$.
\end{IEEEproof}

The next theorem summarizes the performance of Algorithm \ref{alg:with-csit} and follows directly from Lemma \ref{lm:with-csit-dpp-bound} and Lemma \ref{lm:with-csit-deterministic-queue-bound}.

\begin{Thm} \label{thm:with-csit-performance} Fix $\epsilon>0$ and choose $V = \frac{\max\{\bar{P}^2, (P-\bar{P})^2\} }{2\epsilon}$ in Algorithm \ref{alg:with-csit}, then for all $t>0$: 
\begin{align*}
&\frac{1}{t}\sum_{\tau=0}^{t-1} \mathbb{E}[R(\tau)] \geq R^{\text{opt}} - \epsilon -\phi(\delta), \\
&\frac{1}{t}\sum_{\tau=0}^{t-1} \text{tr}(\mathbf{Q}(\tau)) \leq \bar{P} + \frac{(B+\delta)^2\max\{\bar{P}^2, (P-\bar{P})^2\}  + 2\epsilon (P-\bar{P})}{2\epsilon t}, 
\end{align*}
where $\phi(\delta) =2P\sqrt{N_T} (2B+\delta)\delta$ satisfying $\phi(\delta)\rightarrow 0$ as $\delta\rightarrow 0$, i.e., $\phi(\delta)\in O(\delta)$; and $B, \delta, N_T, P$ and $\bar{P}$ are defined in Section \ref{sec:signal-model}. In particular, the average expected utility is within $\epsilon+\phi(\delta)$ of $R^{\text{opt}}$ and the sample path time average power is within $\epsilon$ of its required constraint $\bar{P}$ whenever $t \geq \Omega(\frac{1}{\epsilon^2})$. 
\end{Thm}

\begin{IEEEproof}~\newline
{\bf Proof of the first inequality:} Fix $t>0$.  For all slots $\tau\in\{0,1,\ldots,t-1\}$, Lemma \ref{lm:with-csit-dpp-bound} guarantees that  $\mathbb{E}[R(\tau)] \geq R^{\text{opt}} + \frac{1}{V} \mathbb{E}[\Delta(\tau)] - \frac{1}{2V}\max\{\bar{P}^2, (P-\bar{P})^2\}   - 2P\sqrt{N_T} (2B+\delta)\delta$.

Summing over $\tau \in \{0, \ldots, t-1\}$ and dividing by $t$ gives: 
\begin{align*}
&\frac{1}{t}\sum_{\tau=0}^{t-1} \mathbb{E}[R(\tau)]\\
 \geq& R^{\text{opt}} + \frac{1}{Vt} \sum_{\tau=0}^{t-1}\mathbb{E}[\Delta(\tau)]  -  \frac{1}{2V}\max\{\bar{P}^2, (P-\bar{P})^2\}   \\ &- 2P\sqrt{N_T} (2B+\delta)\delta\\
\overset{(a)}{=}&R^{\text{opt}} + \frac{1}{2Vt} \big([\mathbb{E}[Z^2(t)] -\mathbb{E}[Z^2(0)]\big)  -  \frac{1}{2V}\max\{\bar{P}^2, (P-\bar{P})^2\}  \\&- 2P\sqrt{N_T} (2B+\delta)\delta\\
\overset{(b)}{\geq} &R^{\text{opt}} -  \frac{1}{2V}\max\{\bar{P}^2, (P-\bar{P})^2\}  - 2P\sqrt{N_T} (2B+\delta)\delta \\
\overset{(c)}{=} &R^{\text{opt}} - \epsilon - 2P\sqrt{N_T} (2B+\delta)\delta
\end{align*}
where (a) follows from the definition that $\Delta(t)  = \frac{1}{2} Z^2(t+1) - \frac{1}{2}Z^2(t)$ and by simplifying the telescoping sum $\sum_{\tau=0}^{t-1}\mathbb{E}[\Delta(\tau)]$; (b) follows from $Z(0)=0$ and $Z(t)\geq 0$; and (c) follows by  substituting $V = \frac{1}{2\epsilon}\max\{\bar{P}^2, (P-\bar{P})^2\} $. \newline
{\bf Proof of the second inequality}:  Fix $t>0$. By Lemma \ref{lm:with-csit-constaint-and-queue-relation}, we have 
\begin{align*}
\frac{1}{t} \sum_{\tau=0}^{t-1} \text{tr}(\mathbf{Q}(\tau)) &\leq  \bar{P} + \frac{Z(t)}{t} \\
&\overset{(a)}{=}\bar{P} + \frac{(B+\delta)^2\max\{\bar{P}^2, (P-\bar{P})^2\}  + 2\epsilon (P-\bar{P})}{2\epsilon t} 
\end{align*}
where (a) follows from Lemma  \ref{lm:with-csit-deterministic-queue-bound} and $V = \frac{1}{2\epsilon}\max\{\bar{P}^2, (P-\bar{P})^2\} $.
\end{IEEEproof} 

Theorem \ref{thm:with-csit-performance} provides a sample path guarantee on average power, which is much stronger than the guarantee in \eqref{eq:book2}. It also shows that convergence time to reach an $\epsilon+O(\delta)$ approximate solution is $O(\frac{1}{\epsilon^2})$.

\subsection{Discussion} \label{sec:with-csit-discussion}

 It is shown that $Z(t)$ in the DPP algorithm is ``attracted" to an optimal Lagrangian dual multiplier of an unknown deterministic convex program \cite{HuangNeely13TAC}. In fact, if we have a good guess of this Lagrangian multiplier and initialize $Z(0)$ close to it, then Algorithm \ref{alg:with-csit} has faster convergence. In addition, the performance bounds derived in Theorem \ref{thm:with-csit-performance} are not tightest possible. The proof of Lemma \ref{lm:with-csit-dpp-bound} involves many relaxations to derive bounds that are simple but can still enable Theorem \ref{thm:with-csit-performance} to show the effect of missing CDIT can be made arbitrarily small by choosing the algorithm parameter $V$ properly and the performance degradation of CSIT inaccuracy scales linearly with respect to $\delta$.  In fact, tighter but more complicated bounds are possible by refining the proof of Lemma \ref{lm:with-csit-dpp-bound}.

A heuristic approach to solve problem \eqref{eq:stochastic-with-csit-obj}-\eqref{eq:stochastic-with-csit-set-con} without channel distribution information is to sample the channel for a large number of realizations and use the empirical distribution as an approximate distribution to solve problem \eqref{eq:stochastic-with-csit-obj}-\eqref{eq:stochastic-with-csit-set-con} directly. This approach has three drawbacks:
\begin{itemize}
\item  For a scalar channel, the empirical distribution based on $O(\frac{1}{\epsilon^2})$ realizations is an $\epsilon$ approximation to the true channel distribution with high probability by the Dvoretzky-Kiefer-Wolfowitz inequality \cite{book_ApproximationStatistics}. However, for an $N_R \times N_T$ MIMO channel, the multi-dimensional empirical distribution requires $O( \frac{N^2_T N^2_R}{\epsilon^2})$ samples to achieve an $\epsilon$ approximation of the true channel distribution \cite{Devroye77EmpiricalDistribution}.  Thus, this approach does not scale well with the number of antennas. 
\item  Even if the empirical distribution is accurate, the complexity of solving problem \eqref{eq:stochastic-with-csit-obj}-\eqref{eq:stochastic-with-csit-set-con} based on the empirical distribution is huge if the channel is from a continuous distribution.  This is known as the curse of dimensionality for stochastic optimization due to the large sample size. In contrast, the complexity of Algorithm \ref{alg:with-csit} is independent of the sample space.
\item This approach is an offline method such that a large number of slots are wasted during the channel sampling process. In contrast, Algorithm \ref{alg:with-csit} is an online method with performance guarantees for all slots.
\end{itemize}

Note that even if we assume the distribution of $\mathbf{H}(t)$ is known and $\mathbf{Q}^\ast(\mathbf{H})$ can be computed by solving problem \eqref{eq:stochastic-with-csit-obj}-\eqref{eq:stochastic-with-csit-set-con}, the optimal  policy $\mathbf{Q}^\ast(\mathbf{H})$ in general cannot achieve $R^{\text{opt}}$ and can violate the long term power constraints when only the approximate versions $\widetilde{\mathbf{H}}(t)$ are known.  For example, consider a MIMO fading system with two possible channel realizations $\mathbf{H}_1$ and $\mathbf{H}_2$ with equal probabilities. Suppose the average power constraint is $\bar{P}=5$ and the optimal policy $\mathbf{Q}^\ast(\mathbf{H})$ satisfies $\text{tr}(\mathbf{Q}^\ast(\mathbf{H}_1)) = 8$ and $\text{tr}(\mathbf{Q}^\ast(\mathbf{H}_2)) = 2$. However, if $\widetilde{\mathbf{H}}_1\neq \mathbf{H}_1$ and $\widetilde{\mathbf{H}}_2 \neq \mathbf{H}_2$, it can be hard to decide the transmit covariance based on $\widetilde{\mathbf{H}}_1$ or $\widetilde{\mathbf{H}}_2$ since the associations between $\widetilde{\mathbf{H}}_1$ and $\mathbf{H}_1$ (or between $\widetilde{\mathbf{H}}_2$ and $\mathbf{H}_2$) are unknown. In an extreme case when $\widetilde{\mathbf{H}}_1 = \widetilde{\mathbf{H}}_2 = \mathbf{H}_1$, if the transmitter uses $\mathbf{Q}^\ast(\widetilde{\mathbf{H}}(t))$ at each slot $t$, the average power constraint is violated and hence the transmit covariance scheme is infeasible.  In contrast, Algorithm \ref{alg:with-csit} can attain the performance in Theorem \ref{thm:with-csit-performance} with inaccurate instantaneous CSIT and no CDIT. 

\section{Delayed CSIT case}

Consider the case of delayed and inaccurate CSIT.  At the beginning of each slot $t\in\{0,1,2,\ldots\}$, channel $\mathbf{H}(t)$ is unknown and only quantized channels of previous slots $\widetilde{\mathbf{H}}(\tau), \tau\in\{0,1,\ldots, t-1\}$ are known.   This is similar to the scenario of online optimization where the decision maker selects $x(t)\in\mathcal{X}$ at each slot $t$ to maximize an unknown reward function $f_t(x)$ based on the information of previous reward functions $f_\tau(x(\tau)), \tau\in\{0,1,\ldots, t-1\}$. The goal is to minimize average regret $\frac{1}{t}\max_{x\in \mathcal{X}}\big[\sum_{\tau=0}^{t-1} f_{\tau}(x)\big] - \frac{1}{t}\sum_{\tau=0}^{t-1} f_{\tau}(x({\tau}))$. The best possible average regret of online convex optimization with general convex reward functions is $O(\frac{1}{\sqrt{t}})$ \cite{Zinkevich03ICML,Hazan07ML}.

The situation in the current paper is different from conventional online optimization because at each slot $t$, the rewards of previous slots, i.e., $R(\tau) =\log\det(\mathbf{I} + \mathbf{H}(\tau) \mathbf{Q}(\tau)\mathbf{H}\herm(\tau)), \tau\in\{0,1,\ldots,t-1\}$, are still unknown due to the fact that the reported channels $\widetilde{\mathbf{H}}(\tau)$ are approximate versions. Nevertheless, an online algorithm without using CDIT is developed in Algorithm \ref{alg:no-csit}.

\begin{algorithm}
\caption{Dynamic Transmit Covariance Design with Delayed CSIT}
\label{alg:no-csit}
Let $\gamma>0$ be a constant parameter and $\mathbf{Q}(0)\in \mathcal{Q}$ be arbitrary. At each time $t\in\{1,2,\ldots\}$, observe $\widetilde{\mathbf{H}}(t-1)$ and do the following:
\begin{itemize}
\item Let $ \widetilde{\mathbf{D}}(t-1) = \widetilde{\mathbf{H}}\herm(t-1) (\mathbf{I}_{N_{R}} +  \widetilde{\mathbf{H}}(t-1)\mathbf{Q}(t-1) \widetilde{\mathbf{H}}\herm(t-1))^{-1}  \widetilde{\mathbf{H}}(t-1)$. Choose transmit covariance $\mathbf{Q}(t) = \mathcal{P}_{\widetilde{\mathcal{Q}}} \big[ \mathbf{Q}(t-1) + \gamma \widetilde{\mathbf{D}}(t-1) \big]$, where $\mathcal{P}_{\widetilde{\mathcal{Q}}}[\cdot]$ is the projection onto convex set $\widetilde{\mathcal{Q}}= \{\mathbf{Q} \in \mathbb{S}_{+}^{N_T} : \text{tr}(\mathbf{Q}) \leq \bar{P}\} $.
\end{itemize}
\end{algorithm}

Define $\mathbf{Q}^{\ast}\in \widetilde{\mathcal{Q}}$ as an optimal solution to problem \eqref{eq:stochastic-no-csit-obj}-\eqref{eq:stochastic-no-csit-set-con}, which depends on the (unknown) distribution for $\mathbf{H}(t)$. Define \[R^{\text{opt}}(t) = \log\det(\mathbf{I} + \mathbf{H}(t)\mathbf{Q}^{\ast}\mathbf{H}\herm(t))\] as the utility at slot $t$ attained by $\mathbf{Q}^{\ast}$. 

If the channel feedback is \emph{accurate}, i.e., $\widetilde{\mathbf{H}}(t-1) = \mathbf{H}(t-1), \forall t\in\{1,2,\ldots\}$, then $\widetilde{\mathbf{D}}(t-1)$ is the gradient of $R(t-1)$ at point $Q(t-1)$. Fix $\epsilon>0$ and take $\gamma = \epsilon$. The results in \cite{Zinkevich03ICML} ensure that, regardless of the distribution of $\mathbf{H}(t)$:
\begin{align}
&\frac{1}{t}\sum_{\tau=0}^{t-1} R(\tau) \geq \frac{1}{t}\sum_{\tau=0}^{t-1}R^{\text{opt}}(\tau) - \frac{2\bar{P}^{2}}{\epsilon t} - \frac{N_{R}B^{4}}{2}\epsilon , \forall t> 0 \label{eq:online-1}\\
&\text{tr}(\mathbf{Q}(\tau)) \leq \bar{P}, \forall \tau\in\{0,1,\ldots,t-1\} \label{eq:online-2}
\end{align}

The next subsections show that  the performance of Algorithm \ref{alg:no-csit} with \emph{inaccurate} channels  degrades linearly with respect to channel inaccuracy $\delta$.  If $\delta = 0$, then \eqref{eq:online-1} and \eqref{eq:online-2} are recovered. 
\subsection{Transmit Covariance Updates in Algorithm \ref{alg:no-csit}}
This subsection shows the  $\mathbf{Q}(t)$ selection in Algorithm \ref{alg:no-csit} has an (almost) closed-form solution. 

The projection operator involved  in Algorithm \ref{alg:no-csit} by definition is 
\begin{align}
\min \quad & \frac{1}{2} \Vert \mathbf{Q} - \mathbf{X} \Vert^2_F  \label{eq:no-csit-primal-opt-obj}\\
\text{s.t.} \quad  & \text{tr}(\mathbf{Q}) \leq \bar{P}  \label{eq:no-csit-primal-opt-trace}\\
			 & \mathbf{Q} \in \mathbb{S}^{N_{T}}_+ \label{eq:no-csit-primal-opt-sdp}
\end{align}
where $\mathbf{X} = \mathbf{Q}(t-1) +\gamma \widetilde{\mathbf{D}}(t-1)$ is a Hermitian matrix at each slot $t$. 

Without constraint $\text{tr}(\mathbf{Q})\leq \bar{P}$, the projection of Hermitian matrix $\mathbf{X}$ onto the positive semidefinite cone $\mathbb{S}^{n}_+ $ is simply taking the eigenvalue expansion  of $\mathbf{X}$ and dropping terms associated with negative eigenvalues (see Section 8.1.1. in \cite{book_ConvexOptimization}).  Work \cite{Boyd05LeastSqure} considered the projection onto the intersection of the positive semidefinite cone  $\mathbb{S}^{n}_+ $  and an affine subspace given by $\{\mathbf{Q}: \text{tr}(\mathbf{A}_{i}\mathbf{Q}) = b_{i}, i\in\{1,2,\ldots,p\}, \text{tr}(\mathbf{B}_{j}\mathbf{Q}) \leq d_{j}, j\in\{1,2,\ldots, m\}\}$ and developed the dual-based iterative numerical algorithm to calculate the projection. Problem \eqref{eq:no-csit-primal-opt-obj}-\eqref{eq:no-csit-primal-opt-sdp} is a special case, where the affine subspace is given by $\text{tr}(\mathbf{Q}) \leq \bar{P}$, of the projection considered in \cite{Boyd05LeastSqure}. Instead of solving problem \eqref{eq:no-csit-primal-opt-obj}-\eqref{eq:no-csit-primal-opt-sdp} using numerical algorithms, the next lemma summarizes that problem \eqref{eq:no-csit-primal-opt-obj}-\eqref{eq:no-csit-primal-opt-sdp} has an (almost) closed-form solution.

\begin{Lem}\label{lm:no-csit-transmit-covaraince-update}
Consider SVD $\mathbf{X} = \mathbf{U}\herm\boldsymbol{\Sigma}\mathbf{U}$, where $\mathbf{U}$ is a unitary matrix and $\boldsymbol{\Sigma}$ is a diagonal matrix with entries $\sigma_{1}, \ldots, \sigma_{N_{T}}$.  Then the optimal solution to problem \eqref{eq:no-csit-primal-opt-obj}-\eqref{eq:no-csit-primal-opt-sdp} is given by $\mathbf{Q}^{\ast} = \mathbf{U}\herm\mathbf{\Theta}^{\ast} \mathbf{U}$, where $\mathbf{\Theta}^{\ast}$ is a diagonal matrix with entries $\theta_{1}^\ast, \ldots, \theta_{N_{T}}^{\ast}$ given by, \[\theta_i^\ast = \max[0,\sigma_{i} - \mu^{\ast}], \forall i\in\{1,2,\ldots,N_{T}\},\] where $\mu^{\ast}$ is chosen such that $\sum_{i=1}^{N_T} \theta_i^\ast \leq  \bar{P}$, $\mu^{\ast} \geq 0$ and $\mu^{\ast} \big[\sum_{i=1}^{N_{T}} \theta_i^\ast -  \bar{P}\big] = 0$.  The exact $\mu^\ast$ can be determined using Algorithm \ref{alg:no-csit-primal-opt-exact-alg} with complexity $O(N_T \log N_T)$.
\end{Lem}

\begin{IEEEproof}
See  Appendix \ref{app:pf-lm-no-csit-transmit-covaraince-update} for details.
\end{IEEEproof}

\begin{algorithm} 
\caption{Algorithm to solve problem \eqref{eq:no-csit-primal-opt-obj}-\eqref{eq:no-csit-primal-opt-sdp}} 
\label{alg:no-csit-primal-opt-exact-alg}
\begin{enumerate}
\item Check if $\sum_{i=1}^{N_{T}} \max[0,\sigma_{i}] \leq \bar{P}$ holds. If yes, let $\mu^{\ast} = 0$ and $\theta_{i}^{\ast} =\max[0,\sigma_{i}], \forall i\in\{1,2,\ldots,N_{T}\}$ and terminate the algorithm; else, continue to the next step.
\item Sort all $\sigma_{i}, \in\{1,2,\ldots,N_{T}\}$ in a decreasing order $\pi$ such that $\sigma_{\pi(1)} \geq \sigma_{\pi(2)} \geq \cdots \geq \sigma_{\pi(N_{T})}$. Define $S_{0}=0$.
\item For $i=1$ to $N_{T}$
\begin{itemize}
\item  Let $S_{i} = S_{i-1} + \sigma_{i}$. Let $\mu^{\ast} = \frac{S_{i} -\bar{P} }{i}$. 
\item  If $\mu^{\ast}\geq 0$, $\sigma_{\pi(i)} - \mu^{\ast}>0$ and $\sigma_{\pi(i+1)} - \mu^{\ast} \leq 0$, then terminate the loop; else, continue to the next iteration in the loop. 
\end{itemize}
\item Let $\theta_{i}^{\ast} =\max[0,\sigma_{i} - \mu^{\ast}] , \forall i\in\{1,2,\ldots,N_{T}\}$ and terminate the algorithm.
\end{enumerate}
\end{algorithm}

\subsection{Performance of Algorithm \ref{alg:no-csit}}

Define $\mathbf{D}(t-1) = \mathbf{H}\herm(t-1) (\mathbf{I}_{N_{R}} + \mathbf{H}(t-1)\mathbf{Q}(t-1) \mathbf{H}\herm(t-1))^{-1}  \mathbf{H}(t-1)$, which is the gradient of $R(t-1)$ at point $\mathbf{Q}(t-1)$ and is unknown to the transmitter due to the unavailability of $\mathbf{H}(t-1)$. The next lemma relates $\widetilde{\mathbf{D}}(t-1)$ and $\mathbf{D}(t-1)$.

\begin{Lem}\label{lm:gradient-approx-bound}For all slots $t\in\{1,2,\ldots\}$, we have
\begin{enumerate}
\item $\Vert \mathbf{D}(t-1)\Vert_{F} \leq \sqrt{N_{R}} B^{2}$.
\item $\Vert \mathbf{D}(t-1) -  \widetilde{\mathbf{D}}(t-1) \Vert_{F} \leq  \psi(\delta)$,
where $\psi(\delta) = \big(\sqrt{N_{R}} B +  \sqrt{N_{R}} (B + \delta)  + (B + \delta)^{2} N_{R} \bar{P} ( 2B + \delta)  \big)\delta$ satisfying $\psi(\delta) \rightarrow 0$ as $\delta\rightarrow 0$, i.e., $\psi(\delta) \in O(\delta)$.
\item $\Vert \widetilde{\mathbf{D}}(t-1)\Vert_{F} \leq \psi(\delta)  +\sqrt{N_{R}} B^{2} $
\end{enumerate}
where $B, \delta, N_R,N_T, P$ and $\bar{P}$ are defined in Section \ref{sec:signal-model}
\end{Lem}
\begin{IEEEproof}
See Appendix \ref{app:pf-lm-gradient-approx-bound} for details.
\end{IEEEproof}

The next theorem summarizes the performance of Algorithm \ref{alg:no-csit}.

\begin{Thm}\label{thm:no-csit-performance}
Fix $\epsilon>0$ and define $\gamma = \epsilon$. Under Algorithm \ref{alg:no-csit}, we have\footnote{In our conference version \cite{YuNeely16INFOCOM}, the first inequality of this theorem is mistakenly given by $\frac{1}{t} \sum_{\tau=0}^{t-1} R(\tau) \geq \frac{1}{t}\sum_{\tau=0}^{t-1} R^{\text{opt}}(\tau) - \frac{\bar{P}}{\epsilon t} - \frac{(\psi(\delta)+\sqrt{N_{R}}B^2)^{2}}{2}\epsilon - 2\psi(\delta)\bar{P}$.} for all $t>0$:
\begin{align*}
\frac{1}{t} \sum_{\tau=0}^{t-1} R(\tau) \geq& \frac{1}{t}\sum_{\tau=0}^{t-1} R^{\text{opt}}(\tau) - \frac{2\bar{P}^{2}}{\epsilon t} - \frac{(\psi(\delta)+\sqrt{N_{R}}B^2)^{2}}{2}\epsilon  \\ &- 2\psi(\delta)\bar{P}\\
\text{tr}(\mathbf{Q}(\tau)) \leq& \bar{P}, \forall \tau\in\{0,1,\ldots,t-1\}
\end{align*}
where $\psi(\delta)$ is the constant defined in Lemma \ref{lm:gradient-approx-bound} and $B, \delta,N_R, P$ and $\bar{P}$ are defined in Section \ref{sec:signal-model}. In particular, the sample path time average utility is within $O(\epsilon) + 2\psi(\delta)\bar{P}$ of the optimal average utility for problem \eqref{eq:stochastic-no-csit-obj}-\eqref{eq:stochastic-no-csit-set-con}  whenever $t \geq \frac{1}{\epsilon^2}$.
\end{Thm}
\begin{IEEEproof}
The second inequality trivially follows from the fact that $\mathbf{Q}(t)\in\widetilde{\mathcal{Q}}, \forall t\in\{0,1,\ldots\}$. It remains to prove the first inequality. This proof extends the regret analysis of conventional online convex optimization \cite{Zinkevich03ICML} by considering inexact gradient $\widetilde{\mathbf{D}}(t-1)$. 

For all slots $\tau\in\{1,2,\ldots\}$, the transmit covariance update in Algorithm \ref{alg:no-csit} satisfies:
\begin{align*}
&\Vert \mathbf{Q}(\tau) - \mathbf{Q}^\ast \Vert_{F}^{2}\\ 
=& \Vert  \mathcal{P}_{\widetilde{\mathcal{Q}}} \big[ \mathbf{Q}(\tau-1) + \gamma \widetilde{\mathbf{D}}(\tau-1)\big] - \mathbf{Q}^{\ast} \Vert_{F}^{2}\\
\overset{(a)}{\leq}& \Vert  \mathbf{Q}(\tau-1) +\gamma \widetilde{\mathbf{D}}(\tau-1)- \mathbf{Q}^{\ast}\Vert_{F}^{2} \\
=& \Vert \mathbf{Q}(\tau-1)- \mathbf{Q}^{\ast}\Vert_{F}^{2} + 2\gamma \text{tr} \big(\widetilde{\mathbf{D}}\herm(\tau-1) (\mathbf{Q}(\tau-1) - \mathbf{Q}^{\ast})\big) \\ &+  \gamma^2 \Vert \widetilde{\mathbf{D}}(\tau-1)\Vert_{F}^{2}\\
= &\Vert \mathbf{Q}(\tau-1)- \mathbf{Q}^{\ast}\Vert_{F}^{2} + 2\gamma \text{tr}\big(\mathbf{D}\herm(\tau-1) (\mathbf{Q}(\tau-1) - \mathbf{Q}^{\ast})\big) \\ &+ 2\gamma \text{tr}\big(( \widetilde{\mathbf{D}}(t-1) - \mathbf{D}(\tau-1))\herm (\mathbf{Q}(\tau-1) - \mathbf{Q}^{\ast})\big) \\ &+ \gamma^2 \Vert \widetilde{\mathbf{D}}(\tau-1)\Vert_{F}^{2},
\end{align*}
where (a) follows from the non-expansive property of projections onto convex sets. Define $\Delta(t) = \Vert \mathbf{Q}(t+1) - \mathbf{Q}^{\ast}\Vert_{F}^{2} - \Vert\mathbf{Q}(t) - \mathbf{Q}^{\ast}\Vert_{F}^{2}$. Rearranging terms in the last equation and dividing by factor $2\gamma$ implies \begin{align}
&\text{tr}\big(\mathbf{D}\herm(\tau-1) (\mathbf{Q}(\tau-1) - \mathbf{Q}^{\ast}) \big) \nonumber \\
\geq &\frac{1}{2\gamma}\Delta(\tau-1) - \frac{\gamma}{2} \Vert \widetilde{\mathbf{D}}(\tau-1)\Vert_{F}^{2} \nonumber \\ &- \text{tr}\big(( \widetilde{\mathbf{D}}(\tau-1) - \mathbf{D}(\tau-1))\herm (\mathbf{Q}(\tau-1) - \mathbf{Q}^{\ast})\big)   \label{eq:pf-thm-nocsit-performance-eq1}
\end{align} 
Define $f(\mathbf{Q}) = \log\det(\mathbf{I} + \mathbf{H}(\tau-1) \mathbf{Q} \mathbf{H}\herm(\tau-1))$. By Fact \ref{fact:derivative-log-det} in Appendix \ref{sec:linear-algebra}, $f(\cdot)$ is concave over $\widetilde{\mathcal{Q}}$ and $\mathbf{D}(t-1) = \nabla_{\mathbf{Q}} f(\mathbf{Q}(t-1))$. Note that $\mathbf{Q}^{\ast}\in \widetilde{\mathcal{Q}}$. By Fact \ref{fact:convex-function-first-order-condition} in Appendix \ref{sec:linear-algebra}, we have 
\begin{align}
f(\mathbf{Q}(\tau-1)) - f(\mathbf{Q}^{\ast}) \geq \text{tr}(\mathbf{D}\herm(\tau-1) (\mathbf{Q}(\tau-1) - \mathbf{Q}^{\ast})) \label{eq:pf-thm-nocsit-performance-eq2}
\end{align} 
Note that $f(\mathbf{Q}(\tau-1)) = R(\tau-1)$ and $f(\mathbf{Q}^{\ast}) = R^{\text{opt}}(\tau-1)$. Combining \eqref{eq:pf-thm-nocsit-performance-eq1} and \eqref{eq:pf-thm-nocsit-performance-eq2} yields
\begin{align}
&R(\tau-1) - R^{\text{opt}}(\tau-1) \nonumber \\
\geq &\frac{1}{2\gamma}\Delta(\tau-1) - \frac{\gamma}{2} \Vert \widetilde{\mathbf{D}}(\tau-1)\Vert_{F}^{2}  \nonumber \\ &- \text{tr}\big(( \widetilde{\mathbf{D}}(\tau-1) - \mathbf{D}(\tau-1))\herm (\mathbf{Q}(\tau-1) - \mathbf{Q}^{\ast})\big)  \nonumber\\
\overset{(a)}{\geq}&\frac{1}{2\gamma}\Delta(\tau-1) - \frac{\gamma}{2} \Vert \widetilde{\mathbf{D}}(\tau-1)\Vert_{F}^{2}  \nonumber \\&- \Vert \widetilde{\mathbf{D}}(\tau-1) -  \mathbf{D}(\tau-1) \Vert_{F} \Vert \mathbf{Q}(\tau-1) - \mathbf{Q}^{\ast}\Vert_{F} \nonumber\\
\overset{(b)}{\geq}&\frac{1}{2\gamma}\Delta(\tau-1) - \frac{\gamma}{2} (\psi(\delta)+\sqrt{N_{R}} B^{2})^{2} - 2\psi(\delta)\bar{P} \nonumber
\end{align}
where (a) follows from Fact \ref{fact:1} in Appendix \ref{sec:linear-algebra} and (b) follows from Lemma \ref{lm:gradient-approx-bound} and the fact that  $\Vert \mathbf{Q}(\tau-1) -\mathbf{Q}^\ast\Vert_{F} \leq \Vert \mathbf{Q}(\tau-1)\Vert_{F} + \Vert\mathbf{Q}^\ast\Vert_{F} \leq \text{tr}(\mathbf{Q}(\tau-1)) + \text{tr}(\mathbf{Q}^{\ast}) \leq 2\bar{P}$, which is implied by Fact \ref{fact:1},  Fact \ref{fact:frobenius-trace-inequality} in Appendix \ref{sec:linear-algebra} and the fact that $\mathbf{Q}(\tau-1), \mathbf{Q}^{\ast}\in \widetilde{\mathcal{Q}}$. Replacing $\tau -1$ with $\tau$  yields for all $\tau\in\{0,1,\ldots\}$
\begin{align}
&R(\tau)- R^{\text{opt}}(\tau) \nonumber \\
\geq& \frac{1}{2\gamma}\Delta(\tau) - \frac{\gamma}{2} (\psi(\delta)+\sqrt{N_{R}} B^{2})^{2} - 2\psi(\delta)\bar{P} \label{eq:pf-thm-nocsit-performance-eq3}
\end{align}

Fix $t>0$. Summing over $\tau\in\{0,1,\ldots,t-1\}$, dividing by factor $t$ and simplifying telescope sum $\sum_{\tau=0}^{t-1} \Delta(\tau)$ gives 
\begin{align*}
&\frac{1}{t}\sum_{\tau=0}^{t-1} R(\tau) -  \frac{1}{t}\sum_{\tau=0}^{t-1} R^{\text{opt}}(\tau) )\\ 
\geq & \frac{1}{2\gamma t}(\Vert \mathbf{Q}(t) - \mathbf{Q}^{\ast}\Vert^2_{F} - \Vert\mathbf{Q}(0) - \mathbf{Q}^{\ast}\Vert^{2}_{F}) \\ &- \frac{\gamma}{2} (\psi(\delta)+\sqrt{N_{R}} B^{2})^{2} - 2\psi(\delta)\bar{P} \\
\overset{(a)}{\geq} & -\frac{2\bar{P}^{2}}{\gamma t} - \frac{\gamma}{2} (\psi(\delta) +\sqrt{N_{R}} B^{2})^{2} - 2\psi(\delta)\bar{P}
\end{align*}
where (a) follows from $\Vert \mathbf{Q}(0) - \mathbf{Q}^{\ast}\Vert_{F}\leq \Vert \mathbf{Q}(0)\Vert_{F} +\Vert \mathbf{Q}^{\ast} \Vert_{F} \leq \text{tr}(\mathbf{Q}(0)) + \text{tr}(\mathbf{Q}^{\ast}) \leq 2\bar{P}$ and $\Vert\mathbf{Q}(t) - \mathbf{Q}^{\ast}\Vert^{2}_{F}\geq 0$.
\end{IEEEproof}

Theorem \ref{thm:no-csit-performance} proves a sample path guarantee on the utility. It shows that the convergence time to reach an $O(\epsilon)+2\psi(\delta)\bar{P}$ approximate solution is $1/\epsilon^{2}$. Note that if $\delta =0$, then equations \eqref{eq:online-1} and \eqref{eq:online-2} are recovered by Theorem \ref{thm:no-csit-performance}. Theorem \ref{thm:no-csit-performance} also isolates the effect of missing CDIT and CSIT inaccuracy. The error term $O(\epsilon)$ is corresponding to the effect of missing CDIT and can be made arbitrarily small by choosing a small $\gamma$ and running the algorithm for more than $\frac{1}{\epsilon^2}$ iterations. The observation is that the effect of missing CDIT vanishes as Algorithm \ref{alg:no-csit} runs for a sufficiently long time and hence delayed but accurate CSIT is almost as good as CDIT. The other error term $2\psi(\delta)\bar{P}$ is corresponding to the effect of CSIT inaccuracy and does not vanish. The performance degradation due to channel inaccuracy scales linearly with respect to the channel error since $\psi(\delta)\in O(\delta)$. Intuitively, this is reasonable since any algorithm based on inaccurate CSIT is actually optimizing another different MIMO system. 

\subsection{Extensions}\label{sec:delay-csit-extension}
\subsubsection{$T$-Slot Delayed and Inaccurate CSIT}
Thus far, we have assumed that CSIT is always delayed by one slot. In fact, if CSIT is delayed by $T$ slots, we can modify the update of transmit covariances in Algorithm \ref{alg:no-csit} as $\mathbf{Q}(t) = \mathcal{P}_{\widetilde{\mathcal{Q}}}[\mathbf{Q}(t-T) + \gamma \widetilde{\mathbf{D}}(t-T)]$. A $T$-slot version of Theorem  \ref{thm:no-csit-performance} can be similarly proven.

\subsubsection{Algorithm \ref{alg:no-csit} with Time Varying $\gamma$}  Algorithm \ref{alg:no-csit} can be extended to have time varying step size $\gamma(t) = \frac{1}{\sqrt{t}}$ at slot $t$.  The next lemma shows that such an algorithm can approach an $\epsilon+2\psi(\delta)\bar{P}$ approximate solution with  $O(1/\epsilon^{2})$ iterations.

\begin{Lem}
Fix $\epsilon>0$. If we modify Algorithm \ref{alg:no-csit} by using $\gamma(t) = \frac{1}{\sqrt{t}}$ as the step size $\gamma$ at each slot $t$, then for all $t > 0$:
\begin{align*}
\frac{1}{t}\sum_{\tau=0}^{t-1} R(\tau) \geq& \frac{1}{t}\sum_{\tau=0}^{t-1} R^{\text{opt}}(\tau) - \frac{2\bar{P}^{2}}{\sqrt{t}}  - \frac{1}{\sqrt{t}}(\psi(\delta) +\sqrt{N_{R}} B^{2})^{2}  \\ &- 2\psi(\delta)\bar{P},\\
\frac{1}{t} \sum_{\tau=0}^{t-1} \text{tr}(\mathbf{Q}(\tau)) \leq& \bar{P},
\end{align*} 
where $B, \delta,N_R, P$ and $\bar{P}$ are defined in Section \ref{sec:signal-model}
\end{Lem}

\begin{IEEEproof}
The second inequality again follows from the fact that $\mathbf{Q}(t)\in\widetilde{\mathcal{Q}}, \forall t\in\{0,1,\ldots\}$. It remains to prove the first inequality.  With $\gamma(t) = \frac{1}{\sqrt{t}}$, equation \eqref{eq:pf-thm-nocsit-performance-eq3} in the proof of Theorem \ref{thm:no-csit-performance} becomes $R(\tau) - R^{\text{opt}}(\tau)\geq \frac{1}{2\gamma(\tau+1)}\Delta(\tau)  - \frac{\gamma(\tau+1)}{2}(\psi(\delta) + \sqrt{N_{R}}B^{2})^{2} -2 \psi(\delta) \bar{P}$  for all $\tau\in \{0,1,\ldots\}$.
Fix $t>0$.  Summing over $\tau\in\{0,1,\ldots, t-1\}$ and dividing by factor $t$ yields that for all $t>0$:
\begin{align*}
&\frac{1}{t}\sum_{\tau=0}^{t-1}R(\tau) -\frac{1}{t}\sum_{\tau=0}^{t-1} R^{\text{opt}}(\tau)\\
\geq&  \frac{1}{2t}\sum_{\tau=0}^{t-1}\sqrt{\tau+1} \Delta(\tau) - \frac{1}{t} \left(\sum_{\tau=0}^{t-1}\frac{1}{2\sqrt{\tau+1}}\right)(\psi(\delta)+\sqrt{N_{R}} B^{2})^{2} \\&- 2\psi(\delta)\bar{P} \\
\overset{(a)}{\geq} &  -\frac{2\bar{P}^{2}}{\sqrt {t}} - \frac{1}{\sqrt{t}} (\psi(\delta) +\sqrt{N_{R}} B^{2})^{2} - 2\psi(\delta)\bar{P}
\end{align*}
where (a) follows because $\sum_{\tau=0}^{t-1}\sqrt{\tau+1} \Delta(\tau) 
= \sqrt{t} \Vert \mathbf{Q}(t) - \mathbf{Q}^{\ast}\Vert^{2}_{F} - \Vert \mathbf{Q}(0) - \mathbf{Q}^{\ast}\Vert^{2}_{F} +  \sum_{\tau=0}^{t-2} (\sqrt{\tau+1} - \sqrt{\tau+2}) \Vert \mathbf{Q}(\tau+1) - \mathbf{Q}^{\ast}\Vert^{2}_{F} 
\geq - \Vert \mathbf{Q}(0) - \mathbf{Q}^{\ast}\Vert^{2}_{F} + 4\bar{P}^{2} \sum_{\tau=0}^{t-2} (\sqrt{\tau+1} - \sqrt{\tau+2}) \geq -4\bar{P}^{2}\sqrt{t}$ and $\sum_{\tau=0}^{t-1}\frac{1}{2\sqrt{\tau+1}} \leq \sqrt{t}$.
\end{IEEEproof}

An advantage of time varying step sizes is the performance automatically gets improved as the algorithm runs and there is no need to restart the algorithm with a different constant step size if a better performance is demanded.

\section{Rate Adaptation }

To achieve the capacity characterized by either problem \eqref{eq:stochastic-with-csit-obj}-\eqref{eq:stochastic-with-csit-set-con} or problem \eqref{eq:stochastic-no-csit-obj}-\eqref{eq:stochastic-no-csit-set-con}, we also need to consider the rate allocation associated with the transmit covariance, namely, we need to decide how much data is delivered at each slot.  If the accurate instantaneous CSIT is available, the transmitter can simply deliver $\log\det (\mathbf{I} + \mathbf{H}(t) \mathbf{Q}(t) \mathbf{H}\herm(t))$ amount of data at slot $t$ once $\mathbf{Q}(t)$ is decided.  However, in the cases when instantaneous CSIT is inaccurate or only delayed CSIT is available, the transmitter does not know the associated instantaneous channel capacity without knowing $\mathbf{H}(t)$.  These cases belong to the representative communication scenarios where channels are unknown to the transmitter and \emph{rateless codes} are usually used as a solution. To send $N$ bits of source data,  the rateless code keeps sending encoded information bits without knowing instantaneous  channel capacity such that the receiver can decode all $N$ bits as long as the accumulated channel capacity for sufficiently many slots is larger than $N$. Many practical rateless codes for scalar or MIMO fading channels have been designed in \cite{Erez12IT,Fan10TWC,LiTse16ISIT}.

This section provides an information theoretical rate adaptation policy based on rateless codes that can be combined with the dynamic power allocation algorithms developed in this paper.  

The rate adaptation scheme is as follows: Let $N$ be a large number.  Encode $N$ bits of source data with a capacity achieving code for a channel with capacity no less than $N$ bits per slot. At slot $0$, deliver the above encoded data with transmit covariance $\mathbf{Q}(0)$ given by Algorithm \ref{alg:with-csit} or Algorithm \ref{alg:no-csit}.  The receiver knows channel $\mathbf{H}(0)$,  calculates the channel capacity $R(0) = \log\det(\mathbf{I} + \mathbf{H}(0) \mathbf{Q}(0) \mathbf{H}\herm (0))$; and reports back the scalar $R(0)$ to the transmitter.  At slot $1$, the transmitter removes the first $R(0)$ bits from the $N$ bits of source data,  encodes the remaining $N-R(0)$ bits with a capacity achieving code for a channel with capacity no less than $N-R(0)$ bits per slot; and delivers the encoded data with transmit covariance $\mathbf{Q}(1)$ given by Algorithm \ref{alg:with-csit} or Algorithm \ref{alg:no-csit}.  The receiver knows channel $\mathbf{H}(1)$,  calculates the channel capacity $R(1) = \log\det(\mathbf{I} + \mathbf{H}(1) \mathbf{Q}(1) \mathbf{H}\herm (1))$; and reports back the scalar $R(1)$ to the transmitter. Repeat the above process until slot $T-1$ such that $\sum_{t=0}^{T-1} R(t) >N$. 

For the decoding,  the receiver can decode all the $N$ bits in a reverse order using the idea of successive decoding \cite{book_FundamentalWireless}.  At slot $T-1$, since $N-\sum_{t=1}^{T-2} R(t) < R(T-1)$, that is, $N-\sum_{t=0}^{T-2} R(t) < R(T-1)$ bits of source data are delivered over a channel with capacity $R(T-1)$ bits per slot, the receiver can decode all delivered data ($N-\sum_{t=0}^{T-2} R(t)$ bits) with zero error.  Note that $N-\sum_{t=0}^{T-3} R(t) = R(T-2) + N-\sum_{t=0}^{T-2} R(t)$ bits are delivered at slot $T-2$ over a channel with capacity $R(T-2)$ bits per slot.  The receiver subtracts the $N-\sum_{t=0}^{T-2} R(t)$ bits that are already decoded such that only $R(T-2)$ bits remain to be decoded. Thus, the $R(T-2)$ bits can be successfully decoded.  Repeat this process until all $N$ bits are decoded. 

Using the above rate adaptation and decoding strategy,  $N$ bits are delivered and decoded within $T-1$ slots during which the sum capacity is $\sum_{t=0}^{T-1} R(t)$ bits. When $N$ is large enough, the rate loss $\sum_{t=0}^{T-1} R(t) -N$ is negligible.  This rate adaptation scheme does not require $\mathbf{H}(t)$ and only requires to report back the scalar $R(t-1)$ to the transmitter at each slot $t$.

\section{Simulations}
\subsection{A simple MIMO system with two channel realizations}

Consider a $2\times 2$ MIMO system with two equally likely channel realizations: 
\begin{align*}
\mathbf{H}_1 = \left[\begin{array}{cc}  1.3131e^{j1.9590\pi}  & 2.3880e^{j0.7104\pi} \\ 2.5567e^{j1.5259\pi}  &2.8380e^{j0.3845\pi} \end{array}\right],\\
\mathbf{H}_2 = \left[\begin{array}{cc} 1.4781e^{j0.9674\pi}  &1.5291e^{j0.1396\pi}\\ 0.0601e^{j0.9849\pi} & 0.1842 e^{j1.9126\pi} \end{array}\right].
\end{align*}
This simple scenario is considered as a test case because, when there are only two possible channels with known 
channel probabilities, it is easy to find an optimal baseline algorithm by solving problem \eqref{eq:stochastic-with-csit-obj}-\eqref{eq:stochastic-with-csit-set-con} or problem \eqref{eq:stochastic-no-csit-obj}-\eqref{eq:stochastic-no-csit-set-con} directly.  The goal is to show that the proposed algorithms (which do not 
have channel distribution information) come close to this baseline.  The proposed algorithms can be
implemented just as easily in cases when there are an infinite number of possible channel state matrices, 
rather than just two.  However, in that case it is difficult to find an optimal baseline algorithm since problem \eqref{eq:stochastic-with-csit-obj}-\eqref{eq:stochastic-with-csit-set-con} or problem \eqref{eq:stochastic-no-csit-obj}-\eqref{eq:stochastic-no-csit-set-con} are difficult to solve.\footnote{As discussed in Section \ref{sec:with-csit-discussion}, this is known as the curse of dimensionality for stochastic optimization due to the large sample size.} 

%\footnote{This is known as the curse of dimensionality for stochastic optimization due to the large sample size. That is, even with CDIT, problem \eqref{eq:stochastic-with-csit-obj}-\eqref{eq:stochastic-with-csit-set-con} and problem \eqref{eq:stochastic-no-csit-obj}-\eqref{eq:stochastic-no-csit-set-con} can be numerically hard to solve when the sample size of $\mathbf{H}$ is large. In contrast, the dynamic algorithms proposed in this paper can deal with problems even with an infinite number of samples and the performance guarantees are independent of the sample size.} 

The power constraints are $\bar{P} =2$ and $P = 3$.  If CSIT has error,  $\mathbf{H}_1$ and $\mathbf{H}_2$ are observed as $\widetilde{\mathbf{H}}_1$ and $\widetilde{\mathbf{H}}_2$, respectively. Consider two CSIT error cases.  CSIT Error Case 1:  $\widetilde{\mathbf{H}}_1 = \left[\begin{array}{cc}  1.3131e^{j2\pi}  & 2.3880e^{j0.75\pi} \\ 2.5567e^{j1.5\pi}  &2.8380e^{j0.5\pi} \end{array}\right]$ and $\widetilde{\mathbf{H}}_2 = \left[\begin{array}{cc} 1.4781e^{j1\pi}  &1.5291e^{j0.25\pi}\\ 0.0601e^{j1\pi} & 0.1842 e^{j2\pi} \end{array}\right]$, where the magnitudes are accurate but the phases are rounded to the nearest $\pi/4$ phase; CSIT Error Case 2:  $\widetilde{\mathbf{H}}_1 = \left[\begin{array}{cc}  1.3e^{j2\pi}  & 2.4e^{j0.5\pi} \\ 2.6e^{j1.5\pi}  &2.8e^{j0.5\pi} \end{array}\right]$ and $\widetilde{\mathbf{H}}_2 = \left[\begin{array}{cc} 1.5e^{j1\pi}  &1.5e^{j0\pi}\\ 0 & 0.2 e^{j2\pi} \end{array}\right]$, where the magnitudes are rounded to the first digit after the decimal point and the phases are rounded to the nearest $\pi/2$ phase.

In the instantaneous CSIT case, consider Baseline 1 where the optimal solution $\mathbf{Q}^\ast(\mathbf{H})$ to problem \eqref{eq:stochastic-with-csit-obj}-\eqref{eq:stochastic-with-csit-set-con} is calculated by assuming the knowledge that $\mathbf{H}_1$ and $\mathbf{H}_2$ appear with equal probabilities and $\mathbf{Q}(t) = \mathbf{Q}^\ast(\mathbf{H}(t))$ is used at each slot $t$. Figure \ref{fig:withcsit-simple} compares the performance of Algorithm \ref{alg:with-csit} (with $V=100$) under various CSIT accuracy conditions and Baseline 1. It can be seen that Algorithm \ref{alg:with-csit} has a performance close to that attained by the optimal solution to problem \eqref{eq:stochastic-with-csit-obj}-\eqref{eq:stochastic-with-csit-set-con} requiring channel distribution information. (Note that a larger $V$ gives a even closer performance with a longer convergence time.) It can also be observed that the performance of Algorithm \ref{alg:with-csit} becomes worse as CSIT error gets larger.

\begin{figure}[htbp]
\centering
   \includegraphics[width=0.5\textwidth,height=0.55\textheight,keepaspectratio=true]{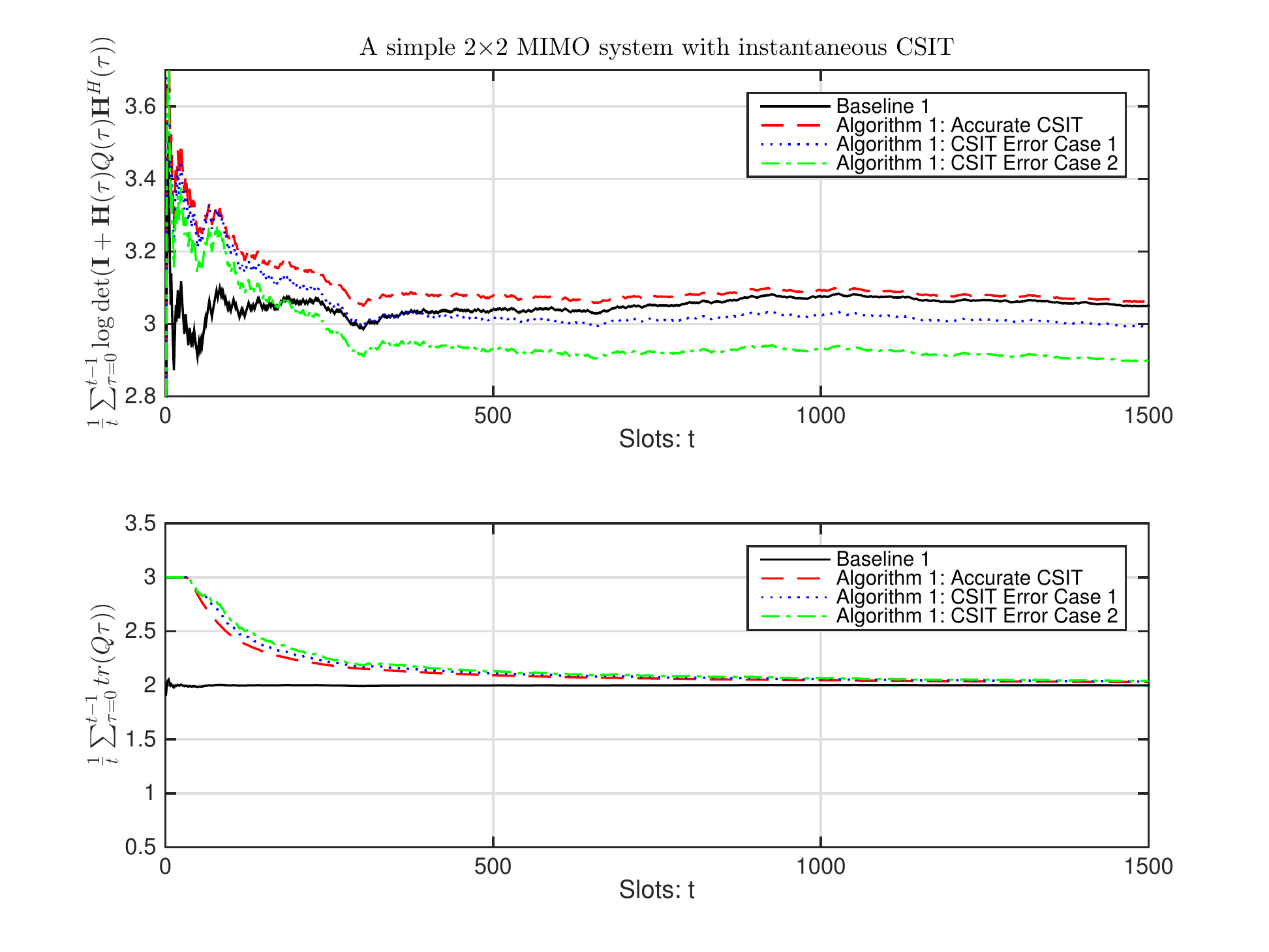} % requires the graphicx package
   \caption{A simple MIMO system with instantaneous CSIT.}
   \label{fig:withcsit-simple}
\end{figure}

In the delayed CSIT case, consider Baseline 2 where the optimal solution $\mathbf{Q}^\ast$ to problem \eqref{eq:stochastic-no-csit-obj}-\eqref{eq:stochastic-no-csit-set-con} is calculated by assuming the knowledge that $\mathbf{H}_1$ and $\mathbf{H}_2$ appear with equal probabilities; and $\mathbf{Q}(t) = \mathbf{Q}^\ast$ is used at each slot $t$. Figure \ref{fig:nocsit-simple} compares the performance of Algorithm \ref{alg:no-csit} (with $\gamma=0.01$) under various CSIT accuracy conditions and Baseline 2. Note that the average power is not drawn since the average power constraint is satisfied for all $t$ in all schemes. It can be seen that Algorithm \ref{alg:no-csit} has a performance close to that attained by the optimal solution to problem \eqref{eq:stochastic-no-csit-obj}-\eqref{eq:stochastic-no-csit-set-con} requiring channel distribution information. (Note that a smaller $\gamma$ gives a even closer performance with a longer convergence time.) It can also be observed that the performance of Algorithm \ref{alg:no-csit} becomes worse as CSIT error gets larger.
\begin{figure}[htbp]
\centering
   \includegraphics[width=0.5\textwidth,height=0.5\textheight,keepaspectratio=true]{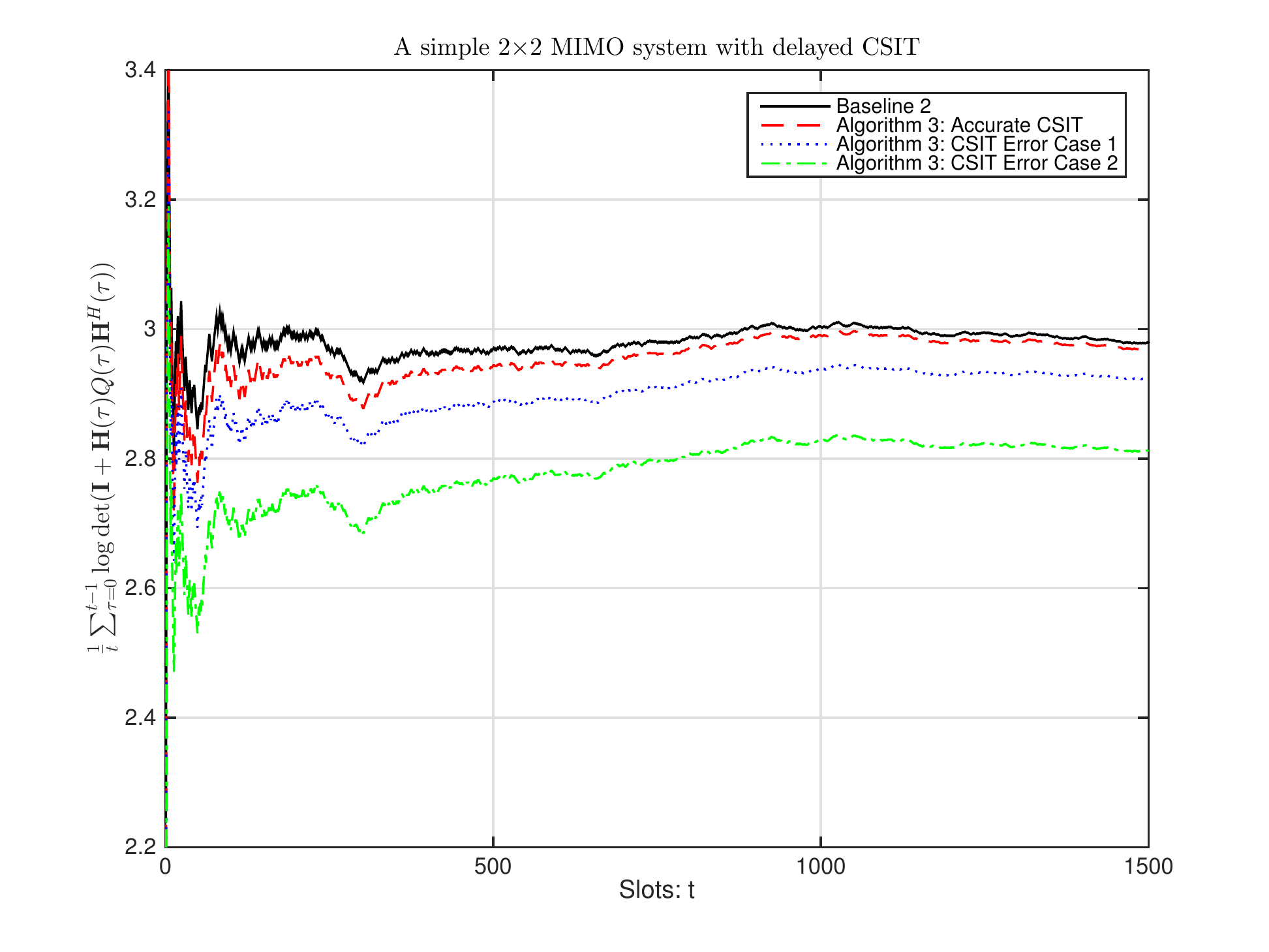} % requires the graphicx package
   \caption{A simple MIMO system with delayed CSIT.}
   \label{fig:nocsit-simple}
\end{figure}  
\subsection{A MIMO system with continuous channel realizations}

This section considers a $2\times 2$ MIMO system with continuous channel realizations. Each entry in $\mathbf{H}(t)$ is equal to $uv$ where $u$ is a complex number whose real part and complex part are standard normal and $v$ is uniform over $[0,0.5]$. In this case, even if the channel distribution information is perfectly known, problem \eqref{eq:stochastic-with-csit-obj}-\eqref{eq:stochastic-with-csit-set-con} and problem \eqref{eq:stochastic-no-csit-obj}-\eqref{eq:stochastic-no-csit-set-con} are infinite dimensional problems and are extremely hard to solve. In practice, to solve the stochastic optimization, people usually approximate the continuous distribution by a discrete distribution with a reasonable number of realizations and solve the approximate optimization that is a large scale deterministic optimization problem. (Baselines 3 and 4 considered below are essentially using this idea.) 

In the instantaneous CSIT case, consider Baseline 3 where we spend $100$ slots to obtain an empirical channel distribution by observing $100$ accurate channel realizations \footnote{By doing so, $100$ slots are wasted without sending any data. The $100$ slots are not counted in the simulation. If they are counted,  Algorithm \ref{alg:with-csit}'s performance advantage over Baseline 3 is even bigger. The delayed CSIT case is similar.}; obtain the optimal solution $\mathbf{Q}^{\ast}(\mathbf{H}), \mathbf{H}\in \mathcal{H}$ to problem \eqref{eq:stochastic-with-csit-obj}-\eqref{eq:stochastic-with-csit-set-con} using the empirical distribution; choose $\mathbf{Q}^\ast({\mathbf{H}})$ where $\mathbf{H} = \argmin_{\mathbf{H}\in \mathcal{H}} \Vert \mathbf{H}-\mathbf{H}(t)\Vert_F$ at each slot $t$.  Figure \ref{fig:withcsit-practical} compares the performance of Algorithm \ref{alg:with-csit} (with $V=100$) and Baseline 3; and shows that Algorithm \ref{alg:with-csit} has a better performance than Baseline 3.
\begin{figure}[htbp]
\centering
   \includegraphics[width=0.5\textwidth,height=0.5\textheight,keepaspectratio=true]{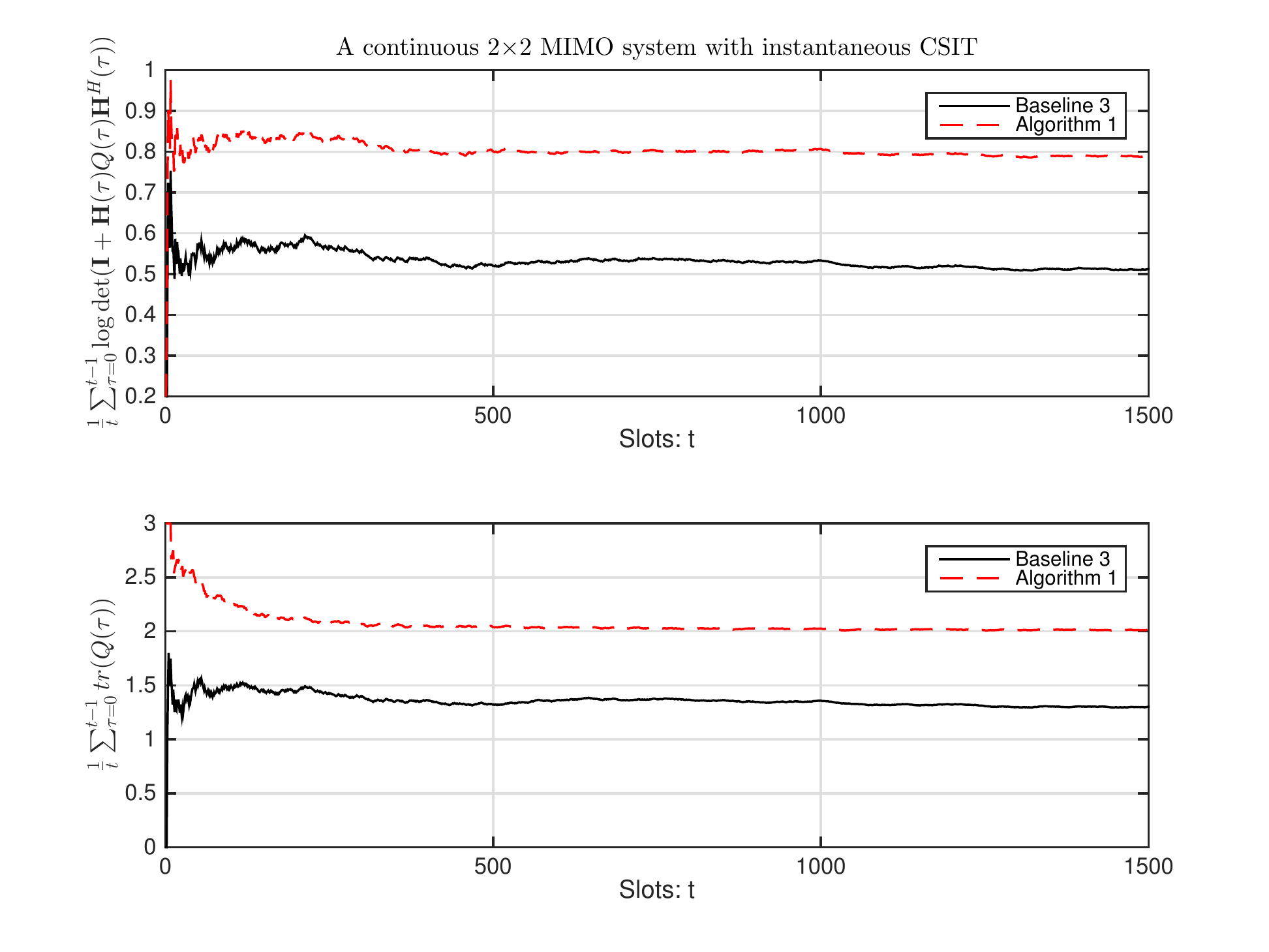} % requires the graphicx package
   \caption{A continuous MIMO system with instantaneous CSIT.}
   \label{fig:withcsit-practical}
\end{figure}

In the delayed CSIT case, consider Baseline 4 where we spend $100$ slots to obtain an empirical channel distribution by observing $100$ accurate channel realizations; obtain the optimal solution $\mathbf{Q}^{\ast}$ to problem \eqref{eq:stochastic-no-csit-obj}-\eqref{eq:stochastic-no-csit-set-con} using the empirical distribution; choose $\mathbf{Q}^\ast$ at each slot $t$.  Figure \ref{fig:nocsit-practical} compares the performance of Algorithm \ref{alg:no-csit} (with $\gamma=0.01$) and Baseline 4; and shows that Algorithm \ref{alg:no-csit} has a better performance than Baseline 4.

\begin{figure}[htbp]
\centering
   \includegraphics[width=0.5\textwidth,height=0.5\textheight,keepaspectratio=true]{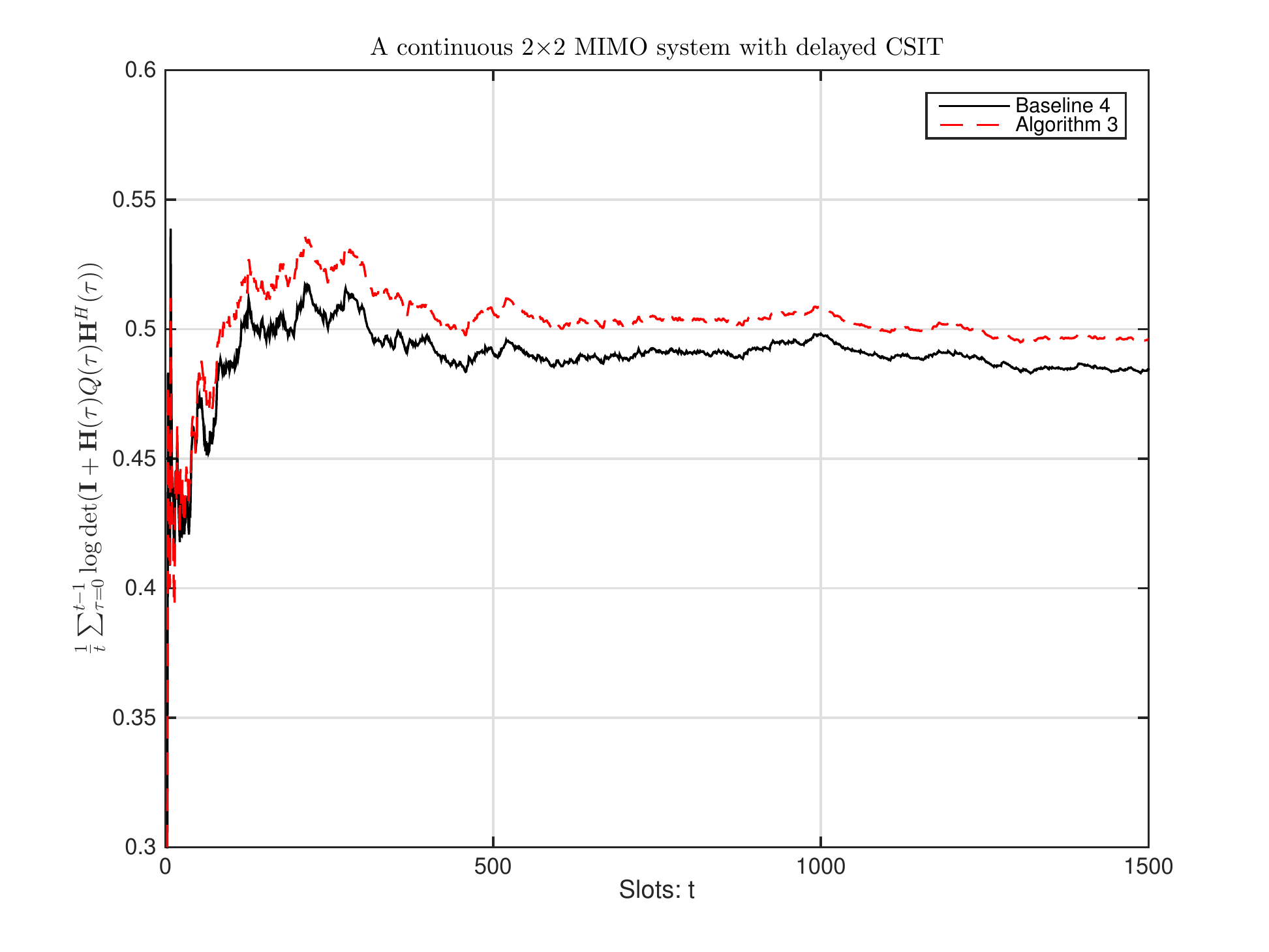} % requires the graphicx package
   \caption{A continuous MIMO system with delayed CSIT.}
   \label{fig:nocsit-practical}
\end{figure}

\section{Conclusion}
This paper considers dynamic transmit covariance design in point-to-point MIMO fading systems without CDIT.  Two different dynamic policies are proposed to deal with the cases of instantaneous CSIT and delayed CSIT, respectively.  In both cases, the proposed dynamic policies can achieve $O(\delta)$ sub-optimality, where $\delta$ is the inaccuracy measure of CSIT.

\appendices

\section{Linear algebra and matrix derivatives} \label{sec:linear-algebra}

\begin{Fact} [\cite{book_MatrixAnalysis}] \label{fact:1} 
For any $\mathbf{A}, \mathbf{B} \in \mathbb{C}^{m\times n}$ and $\mathbf{C}\in \mathbb{C}^{n\times k}$ we have: 
\begin{enumerate}
\item $\Vert \mathbf{A} \Vert_F = \Vert \mathbf{A}\herm \Vert_F = \Vert \mathbf{A}\tran\Vert_F= \Vert -\mathbf{A}\Vert_{F}$.
\item $\Vert \mathbf{A} + \mathbf{B}\Vert_{F} \leq \Vert \mathbf{A}\Vert_{F} + \Vert \mathbf{B}\Vert_{F}$.
\item $\Vert \mathbf{A} \mathbf{C}\Vert_{F} \leq \Vert \mathbf{A}\Vert_{F} \Vert \mathbf{C}\Vert_{F} $.
\item $\vert \text{tr}(\mathbf{A}\herm\mathbf{B})\vert \leq \Vert \mathbf{A}\Vert_{F} \Vert \mathbf{B}\Vert_{F}$.
\end{enumerate}
\end{Fact}
\begin{Fact}[\cite{book_MatrixAnalysis}] \label{fact:frobenius-trace-inequality}
For any $\mathbf{A} \in \mathbb{S}^{n}_{+}$ we have  $\Vert \mathbf{A}\Vert_{F} \leq \text{tr}(\mathbf{A})$. 
\end{Fact}

\begin{Fact}[\cite{Feiten07ISIT}]\label{fact:derivative-log-det}
The function $f:\mathbb{S}_{+}^{n}\rightarrow \mathbb{R}$ defined by $f (\mathbf{Q}) = \log\det(\mathbf{I} + \mathbf{H} \mathbf{Q} \mathbf{H}\herm)$ is concave and its gradient is given by $\nabla_{\mathbf{Q}} f(\mathbf{Q}) = \mathbf{H}\herm (\mathbf{I} + \mathbf{H}\mathbf{Q} \mathbf{H}\herm)^{-1} \mathbf{H}, \forall \mathbf{Q} \in \mathbb{S}^{n}_{+}$.
\end{Fact}
The above fact is developed in \cite{Feiten07ISIT}. A general theory on developing derivatives for functions with complex matrix variables is available in \cite{book_ComplexMatrixDerivatives}. The next fact is the complex matrix version of the first order condition for concave functions of real number variables, i.e., $f(y) \leq f(x) + f^{\prime}(x) (y-x), \forall x,y\in \text{dom} f$ if $f$ is concave.  We also provide a brief proof for this fact.
\begin{Fact}\label{fact:convex-function-first-order-condition}
Let function $f(\mathbf{Q}): \mathbb{S}_{+}^{n} \rightarrow \mathbb{R}$ be a concave function and have gradient $\nabla_{\mathbf{Q}} f(\mathbf{Q}) \in \mathbb{S}^{n}$ at point $\mathbf{Q}$. Then,
$f(\widehat{\mathbf{Q}}) \leq f(\mathbf{Q}) + \text{tr}\big((\nabla_{\mathbf{Q}} f(\mathbf{Q}))\herm(\widehat{\mathbf{Q}} -\mathbf{Q})\big), \forall \widehat{\mathbf{Q}} \in \mathbb{S}^{n}_{+}$.
\end{Fact}
\begin{proof}
Recall that a function is concave if and only if it is concave when restricted to any line along its domain (see page 67 in \cite{book_ConvexOptimization}). For any $\mathbf{Q}, \widehat{\mathbf{Q}} \in \mathbb{S}^{n}_{+}$,  define $g(t) = f(\mathbf{Q} + t(\widehat{\mathbf{Q}} - \mathbf{Q}))$. Thus, $g(t)$ is concave over $[0,1]$; $g(0) = f(\mathbf{Q})$; and $g(1) = f(\widehat{\mathbf{Q}})$. Note that $g^{\prime}(t) = \text{tr}([\nabla_{\mathbf{Q}} f(\mathbf{Q} + t(\widehat{\mathbf{Q}} - \mathbf{Q}))]\herm (\widehat{\mathbf{Q}} - \mathbf{Q}))$ by the chain rule of derivatives when the inner product in complex matrix space $\mathbb{C}^{n\times n}$ is defined as $\langle \mathbf{A}, \mathbf{B} \rangle = \text{tr}(\mathbf{A}\herm\mathbf{B}), \forall \mathbf{A}, \mathbf{B} \in \mathbb{C}^{n\times n}$.  By the first-order condition of concave function $g(t)$, we have  $g(1) \leq g(0) + g^{\prime}(0) (1-0)$.  Note that $ g^{\prime}(0) = \text{tr}([\nabla_{\mathbf{Q}} f(\mathbf{Q})]\herm (\widehat{\mathbf{Q}} - \mathbf{Q}))$. Thus, we have $f(\widehat{\mathbf{Q}}) \leq f(\mathbf{Q}) + \text{tr}\big([\nabla_{\mathbf{Q}} f(\mathbf{Q})]\herm(\widehat{\mathbf{Q}} -\mathbf{Q})\big)$.
\end{proof}

\section{Proof of Lemma \ref{lm:with-csit-transmit-covaraince-update}} \label{app:pf-lm-with-csit-transmit-covaraince-update}

The proof method is an extension of Section 3.2 in \cite{Telatar99MIMOCapacity}, which gives the structure of the optimal transmit covariance in deterministic MIMO channels.

Note that $\log\det(\mathbf{I} + \mathbf{H} \mathbf{Q} \mathbf{H}\herm) \overset{(a)}{=} \log\det(\mathbf{I} + \mathbf{Q} \mathbf{H}\herm\mathbf{H}) 
\overset{(b)}{=} \log\det(\mathbf{I} + \mathbf{Q} \mathbf{U}\herm\boldsymbol{\Sigma} \mathbf{U})
\overset{(c)}{=} \log\det(\mathbf{I} + \boldsymbol{\Sigma}^{1/2}\mathbf{U}\mathbf{Q} \mathbf{U}\herm\boldsymbol{\Sigma}^{1/2})$, where (a) and (c) follows from the elementary identity $\det(\mathbf{I} + \mathbf{A}\mathbf{B}) = \det(\mathbf{I} + \mathbf{B}\mathbf{A}), \forall \mathbf{A}\mathbb{C}^{m\times n}$ and $\mathbf{B} \in \mathbb{C}^{n\times m}$; and (b) follows from the fact that $\mathbf{H}\herm\mathbf{H} = \mathbf{U}\herm\boldsymbol{\Sigma} \mathbf{U}$.  Define $\widetilde{\mathbf{Q}} = \mathbf{U} \mathbf{Q} \mathbf{U}\herm$, which is semidefinite positive if and only if $\mathbf{Q}$ is. Note that $\text{tr}(\widetilde{\mathbf{Q}}) = \text{tr}( \mathbf{U} \mathbf{Q} \mathbf{U}\herm) = \text{tr}(\mathbf{Q})$ by the fact that $\text{tr}(\mathbf{A}\mathbf{B}) = \text{tr}(\mathbf{B}\mathbf{A}), \forall \mathbf{A}\in \mathbb{C}^{m\times n}, \mathbf{B}\in \mathbb{C}^{n\times m}$. Thus, problem \eqref{eq:with-csit-primal-opt-obj}-\eqref{eq:with-csit-primal-opt-sdp} is equivalent to 
\begin{align}
\max_{\widetilde{\mathbf{Q}}} \quad & \log\det(\mathbf{I} + \mathbf{\Sigma}^{1/2} \widetilde{\mathbf{Q}} \mathbf{\Sigma}^{1/2}) - \frac{Z}{V} \text{tr}(\widetilde{\mathbf{Q}}) \label{eq:app-with-csit-primal-opt2-obj}\\
\text{s.t.} \quad  & \text{tr}(\widetilde{\mathbf{Q}}) \leq P  \label{eq:app-with-csit-primal-opt2-trace}\\
			 & \widetilde{\mathbf{Q}} \in \mathbb{S}^{N_{T}}_+  \label{eq:app-with-csit-primal-opt2-sdp}
\end{align}
\begin{Fact}[Hadamard's Inequality, Theorem 7.8.1 in \cite{book_MatrixAnalysis}]
For all $\mathbf{A}\in \mathbb{S}^{n}_{+}$, $\det(\mathbf{A}) \leq \prod_{i=1}^{n} A_{ii}$ with equality if $\mathbf{A}$ is diagonal.
\end{Fact}
The next claim can be proven using Hadamard's inequality.
\begin{Claim}\label{claim:to-prove-thm-with-csit-primal-opt}
Problem \eqref{eq:app-with-csit-primal-opt2-obj}-\eqref{eq:app-with-csit-primal-opt2-sdp} has a diagonal optimal solution. 
\end{Claim}
\begin{IEEEproof}
Suppose problem \eqref{eq:app-with-csit-primal-opt2-obj}-\eqref{eq:app-with-csit-primal-opt2-sdp} has a non-diagonal optimal solution given by matrix $\widetilde{\mathbf{Q}}$. Consider a diagonal matrix $\widehat{\mathbf{Q}}$ whose entries are identical to the diagonal entries of  $\widetilde{\mathbf{Q}}$. Note that $\text{tr}(\widehat{\mathbf{Q}}) = \text{tr}(\widetilde{\mathbf{Q}})$.  To show $\widehat{\mathbf{Q}}$ is a solution no worse than $\widetilde{\mathbf{Q}}$, it suffices to show that $\log\det(\mathbf{I} + \mathbf{\Sigma}^{1/2}\widehat{\mathbf{Q}} \mathbf{\Sigma}^{1/2}) \geq \log\det(\mathbf{I} + \mathbf{\Sigma}^{1/2}\widetilde{\mathbf{Q}} \mathbf{\Sigma}^{1/2})$. This is true becase $\det(\mathbf{I} +  \boldsymbol{\Sigma}^{1/2} \widehat{\mathbf{Q}} \boldsymbol{\Sigma}^{1/2}) = \prod_{i=1}^{N_{T}} (1+\widehat{Q}_{ii} \sigma_{i})
  = \prod_{i=1}^{N_{T}} (1+\widetilde{Q}_{ii}\sigma_{i})\geq \det(\mathbf{I} +  \boldsymbol{\Sigma}^{1/2} \widetilde{\mathbf{Q}} \boldsymbol{\Sigma}^{1/2})$, where the last inequality follows from Hadamard's inequality. Thus, $\widehat{\mathbf{Q}}$ is a solution no worse than $\widetilde{\mathbf{Q}}$ and hence optimal. 
\end{IEEEproof}
By Claim \ref{claim:to-prove-thm-with-csit-primal-opt}, we can consider $\widetilde{\mathbf{Q}} = \Theta= \text{diag}(\theta_{1}, \theta_{2}, \ldots, \theta_{N_{T}})$ and problem \eqref{eq:app-with-csit-primal-opt2-obj}-\eqref{eq:app-with-csit-primal-opt2-sdp} is equivalent to 
\begin{align}
\max_{} \quad & \sum_{i=1}^{N_{T}}\log(1+\theta_i\sigma_{i}) - \frac{Z}{V} \sum_{i=1}^{N_{T}} \theta_i \label{eq:app-with-csit-primal-opt3-obj}\\
\text{s.t.} \quad  & \sum_{i=1}^{N_{T}} \theta_i \leq P  \label{eq:app-with-csit-primal-opt3-trace}\\
			 & \theta_i\geq 0, \forall i\in\{1,2,\ldots, N_{T}\}  \label{eq:app-with-csit-primal-opt3-sdp}
\end{align}
Note that problem \eqref{eq:app-with-csit-primal-opt3-obj}-\eqref{eq:app-with-csit-primal-opt3-sdp} satisfies Slater's condition. So the optimal solution to problem \eqref{eq:app-with-csit-primal-opt3-obj}-\eqref{eq:app-with-csit-primal-opt3-sdp} is characterized by KKT conditions \cite{book_ConvexOptimization}. The remaining part is similar to the derivation of the water-filling solution of power allocation in parallel channels, e.g., the proof of Example 5.2 in \cite{book_ConvexOptimization}. Introducing Lagrange multipliers $\mu\in \mathbb{R}^{+}$ for inequality constraint $\sum_{i=1}^{N_{T}} \theta_i \leq P$ and $\boldsymbol{\nu} = [\nu_1, \ldots, \nu_{N_{T}}]^T\in \mathbb{R}^+$ for inequality constraints $\theta_i\geq 0, i\in\{1,2,\ldots,N_T\}$. Let $\boldsymbol{\theta}^\ast = [\theta_{1}^{\ast}, \ldots, \theta_{N_{T}}^{\ast}]^{T}$ and $(\mu^\ast, \boldsymbol{\nu}^\ast)$ be any primal and dual optimal points with zero duality gap. By the KKT conditions, we have $-\frac{\sigma_{i}}{1+ \theta_i^\ast \sigma_{i}} + Z/V + \mu^{\ast} - \nu_{i}^{\ast}= 0,\forall i\in\{1,2,\ldots,N_{T}\}; 
\sum_{i=1}^{N_{T}} \theta_i^\ast \leq  P;
\mu^{\ast} \geq 0;
\mu^{\ast} \big[\sum_{i=1}^{N_{T}} \theta_i^\ast -  P\big] = 0; 
\theta_i^\ast \geq 0, \forall i\in\{1,2,\ldots,N_{T}\};
\nu_i^\ast \geq 0, \forall i\in\{1,2,\ldots,N_{T}\};
\nu_i^\ast \theta_i^\ast =0, \forall i\in\{1,2,\ldots,N_{T}\}$.

Eliminating $\nu_i^\ast, \forall i\in\{1,2,\ldots,N_{T}\}$ in all equations yields $\mu^{\ast} + Z/V \geq \frac{\sigma_{i}}{1 + \theta_i^{\ast}\sigma_{i}}, \forall i\in\{1,2,\ldots,N_{T}\}; 
\sum_{i=1}^{N_T} \theta_i^\ast \leq  P;
\mu^{\ast} \geq 0;
\mu^{\ast} \big[\sum_{i=1}^{N_{T}} \theta_i^\ast -  P\big] = 0;
\theta_i^\ast \geq 0,\forall i\in\{1,2,\ldots,N_{T}\};
(\mu^{\ast} + Z/V -\frac{\sigma_{i}}{1 + \theta_i^{\ast}\sigma_{i}} ) \theta_i^\ast =0,\forall i\in\{1,2,\ldots,N_{T}\}$.

For all $i\in\{1,2,\ldots,N_{T}\}$, we consider $\mu^\ast + Z/V < \sigma_i$ and $\mu^\ast + Z/V \geq \sigma_i$ separately:
\begin{enumerate}
\item If $\mu^\ast + Z/V< \sigma_i$, then $\mu^{\ast} + Z/V \geq \frac{\sigma_{i}}{1 + \theta_i^{\ast}\sigma_{i}}$ holds only when $\theta_i^\ast >0$, which by $(\mu^{\ast} + Z/V -\frac{\sigma_{i}}{1 + \theta_i^{\ast}\sigma_{i}} ) \theta_i^\ast$ implies that $\mu^{\ast} + Z/V -\frac{\sigma_{i}}{1 + \theta_i^{\ast}\sigma_{i}}=0$, i.e., $\theta_i^{\ast} = \frac{1}{\mu^{\ast} + Z/V} - \frac{1}{\sigma_{i}}$.
\item If $\mu^\ast + Z/V \geq \sigma_i$, then $\theta_i^\ast >0$ is impossible, because $\theta_i^\ast >0$ implies that $\mu^{\ast} + Z/V -\frac{\sigma_{i}}{1 + \theta_i^{\ast}\sigma_{i}} >0$, which together with $\theta_i^\ast >0$ contradict the slackness condition $(\mu^{\ast} + Z/V -\frac{\sigma_{i}}{1 + \theta_i^{\ast}\sigma_{i}} ) \theta_i^\ast =0$. Thus, if $\mu^\ast + Z/V \geq \sigma_i$, we must have $\theta_i^\ast = 0$.
\end{enumerate}
Summarizing both cases, we have $\theta_i^\ast = \max\big[0,\frac{1}{\mu^{\ast} + Z/V} - \frac{1}{\sigma_{i}}\big], \forall i\in\{1,2,\ldots,N_{T}\}$, where $\mu^{\ast}$ is chosen such that $\sum_{i=1}^{n} \theta_i^\ast \leq  P$, $\mu^{\ast} \geq 0$ and $\mu^{\ast} \big[\sum_{i=1}^{N_{T}} \theta_i^\ast -  P\big] = 0$.  

To find such $\mu^{\ast}$, we first check if $\mu^{\ast} =0$. If $\mu^{\ast} =0$ is true, the slackness condition $\mu^{\ast} \big[\sum_{i=1}^{N_{T}} \theta_i^\ast -  P\big] = 0$ is guaranteed to hold and we need to further require $\sum_{i=1}^{N_{T}} \theta_i^{\ast} = \sum_{i=1}^{N_{T}}\max\big[0,\frac{1}{\mu^{\ast} + Z/V} - \frac{1}{\sigma_{i}}\big]\leq P$. Thus $\mu^{\ast} =0$ if and only if $\sum_{i=1}^{N_{T}}  \max\big[0,\frac{1}{Z/V} - \frac{1}{\sigma_{i}}\big] \leq P$. Thus, Algorithm \ref{alg:with-csit-primal-opt-exact-alg} checks if $\sum_{i=1}^{N_{T}}  \max\big[0,\frac{1}{Z/V} - \frac{1}{\sigma_{i}}\big]\leq P$ holds at the first step and if this is true, then we conclude $\mu^{\ast} = 0$ and we are done! 

Otherwise, we know $\mu^{\ast} >0$.  By the slackness condition $\mu^{\ast} \big[\sum_{i=1}^{N_{T}} \theta_{i}^\ast -  P\big] = 0$, we must have $\sum_{i=1}^{N_{T}} \theta_{i}^\ast = \sum_{i=1}^{N_{T}}  \max\big[0,\frac{1}{\mu^{\ast} + Z/V} - \frac{1}{\sigma_{i}}\big] =P$. To find $\mu^{\ast}>0$ such that  $\sum_{i=1}^{N_{T}}  \max\big[0,\frac{1}{\mu^{\ast} + Z/V} - \frac{1}{\sigma_{i}}\big]=P$, we could apply a bisection search by noting that all $\theta_{i}^{\ast}$ are decreasing with respect to $\mu^{\ast}$. 

Another algorithm of finding $\mu^{\ast}$ is inspired by the observation that if $\sigma_{j} \geq \sigma_{k}, \forall j,k\in\{1,2,\ldots, N_{T}\}$, then $\theta_{j}^{\ast} \geq \theta_{k}^{\ast}$. Thus, we first sort all $\sigma_{i}$ in a decreasing order, say $\pi$ is the permutation such that $\sigma_{\pi(1)} \geq \sigma_{\pi(2)} \geq \cdots \geq \sigma_{\pi(N_{T})}$; and then sequentially check if $i\in\{1,2,\ldots,N_{T}\}$ is the index such that $\sigma_{\pi(i)} - \mu^{\ast} \geq 0$ and $\sigma_{\pi(i+1)} - \mu^{\ast} \leq 0$. To check this, we first assume $i$ is indeed such an index and solve the equation $\sum_{j=1}^{i} \big[ \frac{1}{\mu^{\ast} + Z/V} - \frac{1}{\sigma_{\pi(j)}} \big] = P$ to obtain $\mu^{\ast}$; (Note that in Algorithm \ref{alg:with-csit-primal-opt-exact-alg}, to avoid recalculating the partial sum $\sum_{j=1}^{i} \frac{1}{\sigma_{\pi(j)}}$ for each $i$, we introduce the parameter $S_i = \sum_{j=1}^{i} \frac{1}{\sigma_{\pi(j)}}$ and update $S_{i}$ incrementally. By doing this, the complexity of each iteration in the loop is only $O(1)$.) then verify the assumption by checking if $\frac{1}{\mu^{\ast} + Z/V} - \frac{1}{\sigma_{\pi(i)}}  \geq 0$ and $\frac{1}{\mu^{\ast} + Z/V} - \frac{1}{\sigma_{\pi(i+1)}}  \leq 0$. This algorithm is described in Algorithm \ref{alg:with-csit-primal-opt-exact-alg}. 

\section{Proof of Lemma \ref{lm:with-csit-dpp-bound}} \label{app:pf-lm-with-csit-dpp-bound}

\begin{Fact}\label{fact:identiy-plus-sdp-inverse-frobenius-bound}
For all $\mathbf{X}\in \mathbb{S}^{n}_{+}$, we have $\Vert (\mathbf{I} + \mathbf{X})^{-1}\Vert_{F} \leq \sqrt{n}$.
\end{Fact}
\begin{IEEEproof}
Since $\mathbf{X}\in \mathbb{S}^{n}_{+}$, matrix $\mathbf{X}$ has SVD $\mathbf{X} = \mathbf{U}\herm\mathbf{\Sigma}\mathbf{U}$, where $\mathbf{U}$ is unitary and $\mathbf{\Sigma}$ is diagonal with non-negative entries $\sigma_{1},\ldots, \sigma_{n}$. Then $\mathbf{Y} = (\mathbf{I} + \mathbf{X})^{-1} = \mathbf{U}\herm \text{diag}(\frac{1}{1+\sigma_{1}}, \ldots, \frac{1}{1+\sigma_{n}})\mathbf{U}$ is Hermitian. Thus, $\Vert (\mathbf{I} + \mathbf{X})^{-1} \Vert_{F} = \sqrt{\text{tr}(\mathbf{Y}^{2})} = \sqrt{
\sum_{i=1}^{n}(\frac{1}{1+\sigma_{i}})^{2}} \leq \sqrt{n}$.
\end{IEEEproof}

\begin{Fact}\label{fact:HH-difference-bound}
For any $\mathbf{H}, \widetilde{\mathbf{H}}\in \mathbb{C}^{N_R\times N_T}$ with $\Vert \mathbf{H}\Vert_F\leq B$ and $\Vert \mathbf{H} - \widetilde{\mathbf{H}}\Vert_F \leq \delta$, we have $\Vert \mathbf{H}\herm \mathbf{H} - \widetilde{\mathbf{H}}\herm \widetilde{\mathbf{H}}\Vert_F\leq (2B+\delta)\delta$.
\end{Fact}
\begin{IEEEproof}
\begin{align*}
&\Vert \mathbf{H}\herm \mathbf{H} - \widetilde{\mathbf{H}}\herm \widetilde{\mathbf{H}}\Vert_F \\
\overset{(a)}{\leq}& \Vert \mathbf{H}\herm \mathbf{H} - \mathbf{H}\herm \widetilde{\mathbf{H}}\Vert_F +   \Vert \mathbf{H}\herm \widetilde{\mathbf{H}} - \widetilde{\mathbf{H}}\herm \widetilde{\mathbf{H}}\Vert_F\\ \overset{(b)}{\leq} & \Vert \mathbf{H}\herm\Vert_F \Vert \mathbf{H} - \widetilde{\mathbf{H}}\Vert_F + \Vert\mathbf{H}\herm - \widetilde{\mathbf{H}}\herm\Vert_F \Vert \widetilde{\mathbf{H}}\Vert_F\\
\overset{(c)}{\leq} & \Vert \mathbf{H}\herm\Vert_F \Vert \mathbf{H} - \widetilde{\mathbf{H}}\Vert_F + \Vert\mathbf{H}\herm - \widetilde{\mathbf{H}}\herm\Vert_F  \big(\Vert \widetilde{\mathbf{H}} - \mathbf{H}\Vert_F + \Vert \mathbf{H}\Vert_F\big) \\
\leq & 2B\delta + \delta^2
\end{align*}
where (a) and (c) follow from part (2) of Fact \ref{fact:1}; and (b) follows from part (3) of Fact \ref{fact:1}.
\end{IEEEproof}

Fix $Z(t)$ and $V$. Define $\phi(\mathbf{Q}, \mathbf{H}) = V \log \det(\mathbf{I} + \mathbf{H}\mathbf{Q} \mathbf{H}\herm) - Z(t) \text{tr}(\mathbf{Q})$ and $\psi(\mathbf{L},\mathbf{T}) = V \log \det(\mathbf{I} + \mathbf{L} \mathbf{T}\mathbf{L}\herm) - Z(t) \text{tr}(\mathbf{L}\herm\mathbf{L})$.  

\begin{Fact} \label{fact:concave-wrt-channel}
Let $\mathbf{Q}\in \mathbb{S}^{N_T}_+$ have Cholesky decomposition $\mathbf{Q} = \mathbf{L}\herm\mathbf{L}$. Then, $\phi(\mathbf{Q}, \mathbf{H})  = V \log \det(\mathbf{I} + \mathbf{L} \mathbf{T}\mathbf{L}\herm) - Z(t) \text{tr}(\mathbf{L}\herm\mathbf{L}) = \psi(\mathbf{L},\mathbf{T})$ with $\mathbf{T} = \mathbf{H}\herm \mathbf{H}$. Moreover, if $\mathbf{L}$ is fixed, then $\psi(\mathbf{L},\mathbf{T})$ is concave with respect to $\mathbf{T}$ and has gradient $\nabla_\mathbf{T} \psi(\mathbf{L},\mathbf{T}) = V \mathbf{L}^H (\mathbf{I} + \mathbf{L}\mathbf{T}\mathbf{L}^H)^{-1} \mathbf{L}$.
\end{Fact}
\begin{IEEEproof}
Note that 
\begin{align*}
 &V \log \det(\mathbf{I} + \mathbf{H}\mathbf{Q} \mathbf{H}\herm) - Z(t) \text{tr}(\mathbf{Q})\\
=&V \log \det(\mathbf{I} + \mathbf{H}\mathbf{L}\herm\mathbf{L} \mathbf{H}\herm) - Z(t) \text{tr}(\mathbf{L}\herm\mathbf{L})\\
 \overset{(a)}{=} &V \log \det(\mathbf{I} + \mathbf{L} \mathbf{H}\herm\mathbf{H}\mathbf{L}\herm) - Z(t) \text{tr}(\mathbf{L}\herm\mathbf{L})\\
 \overset{(b)}{=} &V \log \det(\mathbf{I} + \mathbf{L} \mathbf{T}\mathbf{L}\herm) - Z(t) \text{tr}(\mathbf{L}\herm\mathbf{L})\\
 =&\psi(\mathbf{L},\mathbf{T})
\end{align*}
where (a) follows from the elementary identity $\det(\mathbf{I} + \mathbf{A}\mathbf{B}) = \det(\mathbf{I} + \mathbf{B}\mathbf{A})$ for any $\mathbf{A}\in \mathbb{C}^{m\times n}$ and $\mathbf{B}\in \mathbb{C}^{n\times m}$; and (b) follows from the definition $\mathbf{T} = \mathbf{H}\herm \mathbf{H}$.

Note that if $\mathbf{L}$ is fixed, then $Z(t) \text{tr}(\mathbf{L}\herm\mathbf{L})$ is a constant. It follows from Fact \ref{fact:derivative-log-det} that $\psi(\mathbf{L},\mathbf{T})$ is concave with respect to $\mathbf{T}$ and has gradient $\nabla_\mathbf{T} \psi(\mathbf{L},\mathbf{T}) = V \mathbf{L}^H (\mathbf{I} + \mathbf{L}\mathbf{T}\mathbf{L}^H)^{-1} \mathbf{L}$. 
\end{IEEEproof}

Let $\mathbf{Q}^\ast(\mathbf{H})$ be an optimal solution to problem \eqref{eq:stochastic-with-csit-obj}-\eqref{eq:stochastic-with-csit-set-con}.  Note that $\mathbf{Q}^\ast(\mathbf{H})$ is a mapping from channel states to transmit covariances and $R^{\text{opt}} = \mathbb{E}[\log\det(\mathbf{I} + \mathbf{H}\mathbf{Q}^\ast(\mathbf{H})\mathbf{H}\herm)]$.  To simplify notation, we denote $\mathbf{Q}^\ast(t) = \mathbf{Q}^\ast(\mathbf{H}(t))$, i.e. the transmit covariance at slot $t$ selected  according to $\mathbf{Q}^\ast(\mathbf{H})$. The next lemma relates the performance of Algorithm \ref{alg:with-csit} and $\mathbf{Q}^{\ast}$ at each slot $t$. 

\begin{Lem}\label{lm:with-csit-relate-omega-only}
Let $\mathbf{Q}(t)$ be yielded by Algorithm \ref{alg:with-csit}. At each slot $t$, we have  $V\log\det(\mathbf{I} + \mathbf{H}(t)\mathbf{Q}(t)\mathbf{H}\herm(t)) - Z(t) \text{tr}(\mathbf{Q}(t)) 
\geq V\log\det(\mathbf{I} + \mathbf{H}(t)\mathbf{Q}^\ast(t)\mathbf{H}\herm(t)) - Z(t) \text{tr}(\mathbf{Q}^\ast(t)) - 2VP\sqrt{N_T} (2B+\delta) \delta$.
\end{Lem}
\begin{IEEEproof}
Fix $t>0$. Let $\widetilde{\mathbf{H}}(t)\in \mathbb{C}^{N_R\times N_T}$ be the observed (inaccurate) CSIT satisfying $\Vert \mathbf{H}(t) - \widetilde{\mathbf{H}}(t)\Vert_F \leq \delta$. The main proof of this lemma can be decomposed into 3 steps:

{\bf $\bullet$ Step 1:} Show that $\phi(\mathbf{Q}(t), \mathbf{H}(t)) \geq \phi(\mathbf{Q}(t), \widetilde{\mathbf{H}}(t)) - VP\sqrt{N_T} (2B+\delta) \delta$. Let $\mathbf{Q}(t) = \mathbf{L}\herm(t)\mathbf{L}(t)$ be an Cholesky decomposition. Define $\mathbf{T}(t) = \mathbf{H}\herm(t) \mathbf{H}(t)$ and $\widetilde{\mathbf{T}}(t) = \widetilde{\mathbf{H}}\herm (t) \widetilde{\mathbf{H}}(t)$.  By Fact \ref{fact:concave-wrt-channel}, we have $\psi(\mathbf{L}(t),\mathbf{T}(t)) = \phi(\mathbf{Q}(t), \mathbf{H}(t))$ and $\psi(\mathbf{L}(t),\widetilde{\mathbf{T}}(t)) = \phi(\mathbf{Q}(t), \widetilde{\mathbf{H}}(t))$; and $\psi$ is concave with respect to $\mathbf{T}$.  By Fact \ref{fact:convex-function-first-order-condition}, we have 
\begin{align*}
&\psi(\mathbf{L}(t),\mathbf{T}(t)) \\
\geq &\psi(\mathbf{L}(t),\widetilde{\mathbf{T}}(t)) - \text{tr} \big( [\nabla_\mathbf{T} \psi(\mathbf{L}(t),\mathbf{T}(t)) ]\herm (\widetilde{\mathbf{T}}(t)-\mathbf{T}(t)\big)\\
\overset{(a)}{\geq} &  \psi(\mathbf{L}(t),\widetilde{\mathbf{T}}(t))  - \Vert \nabla_\mathbf{T} \psi(\mathbf{L}(t),\mathbf{T}(t))\Vert_F \Vert \widetilde{\mathbf{T}}(t)-\mathbf{T}(t)\Vert_F\\
\overset{(b)}{\geq} &  \psi(\mathbf{L}(t),\widetilde{\mathbf{T}}(t)) \\ &- V \Vert \mathbf{L}\herm (t)(\mathbf{I} + \mathbf{L}(t)\mathbf{T}(t)\mathbf{L}\herm(t))^{-1} \mathbf{L}(t)\Vert_F (2B+\delta)\delta\\
\overset{(c)}{\geq} &\psi(\mathbf{L}(t),\widetilde{\mathbf{T}}(t))  - V P\sqrt{N_T}(2B+\delta)\delta
\end{align*}
where (a) follows from part (4) in Fact \ref{fact:1}; (b) follows from $\nabla_\mathbf{T} \psi(\mathbf{L}(t),\mathbf{T}(t))= V\mathbf{L}\herm (t)(\mathbf{I} + \mathbf{L}(t)\mathbf{T}(t)\mathbf{L}\herm(t))^{-1} \mathbf{L}(t)$ by Fact \ref{fact:concave-wrt-channel} and $\Vert \widetilde{\mathbf{T}}(t)-\mathbf{T}(t)\Vert_F \leq \delta(2B+\delta)$ which is further implied by Fact \ref{fact:HH-difference-bound}; and (c) follows from $ \Vert \mathbf{L}\herm (t)(\mathbf{I} + \mathbf{L}(t)\mathbf{T}(t)\mathbf{L}\herm (t))^{-1} \mathbf{L}(t)\Vert_F\leq \Vert \mathbf{L}\herm(t)\Vert_F^2 \Vert (\mathbf{I} + \mathbf{L}(t)\mathbf{T}(t)\mathbf{L}\herm(t))^{-1} \Vert_F \leq P \sqrt{N_T}$ where the first inequality follows from Fact \ref{fact:1} and the second inequality follows from $\Vert \mathbf{L}(t)\Vert_F = \sqrt{\text{tr}(\mathbf{L}\herm(t) \mathbf{L}(t))} =\sqrt{\text{tr}(\mathbf{Q}(t))} \leq \sqrt{P}$ and Fact \ref{fact:identiy-plus-sdp-inverse-frobenius-bound}.

{\bf $\bullet$ Step 2:} Show that $ \phi(\mathbf{Q}(t), \widetilde{\mathbf{H}}(t)) \geq \phi(\mathbf{Q}^\ast(t), \widetilde{\mathbf{H}}(t))$.  This step simply follows from the fact that Algorithm \ref{alg:with-csit} choses $\mathbf{Q}(t)$ to maximize $\phi(\mathbf{Q}, \widetilde{\mathbf{H}}(t)) = V\log\det(\mathbf{I} + \widetilde{\mathbf{H}}(t) \mathbf{Q} \widetilde{\mathbf{H}}\herm(t)) - Z(t) \text{tr}(\mathbf{Q})$ and hence $\mathbf{Q}(t)$ should be no worse than $\mathbf{Q}^\ast(t)$.

{\bf $\bullet$ Step 3:} Show that $\phi(\mathbf{Q}^\ast(t), \widetilde{\mathbf{H}}(t)) \geq \phi(\mathbf{Q}^\ast(t), \mathbf{H}(t)) - VP\sqrt{N_T} (2B+\delta) \delta$. This step is similar to step 1.  Let $\mathbf{Q}^\ast(t) = \mathbf{M}\herm(t)\mathbf{M}(t)$ be an Cholesky decomposition. Define $\mathbf{T}(t) = \mathbf{H}\herm(t) \mathbf{H}(t)$ and $\widetilde{\mathbf{T}}(t) = \widetilde{\mathbf{H}}\herm (t) \widetilde{\mathbf{H}}(t)$.  By Fact \ref{fact:concave-wrt-channel}, we have $\psi(\mathbf{M}(t),\mathbf{T}(t)) = \phi(\mathbf{Q}^\ast(t), \mathbf{H}(t))$ and $\psi(\mathbf{M}(t), \widetilde{\mathbf{T}}(t)) = \phi(\mathbf{Q}^\ast(t), \widetilde{\mathbf{H}}(t))$; and $\psi$ is concave with respect to $\mathbf{T}$.    By Fact \ref{fact:convex-function-first-order-condition}, we have 
\begin{align*}
&\psi( \mathbf{M}(t), \widetilde{\mathbf{T}}(t)) \\
\geq &\psi( \mathbf{M}(t),\mathbf{T}(t)) - \text{tr} \big( [\nabla_\mathbf{T} \psi(\mathbf{M}(t)),\widetilde{\mathbf{T}}(t) ]\herm [\mathbf{T}(t) -\widetilde{\mathbf{T}}(t)]\big)\\
\overset{(a)}{\geq} & \psi( \mathbf{M}(t),\mathbf{T}(t))  - \Vert \nabla_\mathbf{T} \psi(\mathbf{M}(t),\widetilde{\mathbf{T}}(t))\Vert_F \Vert\mathbf{T}(t) -\widetilde{\mathbf{T}}(t)\Vert_F\\
\overset{(b)}{\geq} &  \psi( \mathbf{M}(t),\mathbf{T}(t)) \\ &- V\Vert \mathbf{M}\herm (t)(\mathbf{I} + \mathbf{M}(t)\widetilde{\mathbf{T}}(t)\mathbf{M}\herm(t))^{-1} \mathbf{M}(t)\Vert_F (2B+\delta)\delta\\
\overset{(c)}{\geq} &\psi( \mathbf{M}(t),\mathbf{T}(t)) - V P\sqrt{N_T}(2B+\delta)\delta
\end{align*}
where (a) follows from part (4) in Fact \ref{fact:1}; (b) follows from $\nabla_\mathbf{T} \psi( \mathbf{M}(t), \widetilde{\mathbf{T}}(t))= V\mathbf{M}\herm (t)(\mathbf{I} + \mathbf{M}(t)\widetilde{\mathbf{T}}(t)\mathbf{M}\herm(t))^{-1} \mathbf{M}(t)$ by Fact \ref{fact:concave-wrt-channel} and $\Vert\mathbf{T}(t) - \widetilde{\mathbf{T}}(t)\Vert_F \leq \delta(2B+\delta)$ which is further implied by Fact \ref{fact:HH-difference-bound}; and (c) follows from $ \Vert \mathbf{M}\herm (t)(\mathbf{I} + \mathbf{M}(t)\widetilde{\mathbf{T}}(t)\mathbf{L}\herm (t))^{-1} \mathbf{M}(t)\Vert_F\leq \Vert \mathbf{M}\herm(t)\Vert_F^2 \Vert (\mathbf{I} + \mathbf{M}(t)\widetilde{\mathbf{T}}(t)\mathbf{M}\herm(t))^{-1} \Vert_F \leq P \sqrt{N_T}$ where the first inequality follows from Fact \ref{fact:1} and the second inequality follows from $\Vert \mathbf{M}(t)\Vert_F = \sqrt{\text{tr}(\mathbf{M}\herm(t) \mathbf{M}(t))} =\sqrt{\text{tr}(\mathbf{Q}^\ast(t))} \leq \sqrt{P}$ and Fact \ref{fact:identiy-plus-sdp-inverse-frobenius-bound}.

Combining the above steps yields $\phi(\mathbf{Q}(t), \mathbf{H}(t)) \geq \phi(\mathbf{Q}^\ast(t), \mathbf{H}(t)) - 2VP\sqrt{N_T} (2B+\delta) \delta$.
\end{IEEEproof}
\begin{Lem}
At each time $t\in\{0,1,2,\ldots\}$, we have 
\begin{align}
-\Delta(t) \geq -Z(t) \big(\text{tr}(\mathbf{Q}(t)) - \bar{P}\big) - \frac{1}{2}\max\{\bar{P}^2, (P-\bar{P})^2\}. \label{eq:drift-bound}
\end{align}
\end{Lem}
\begin{IEEEproof}
Fix $t\in\{0,1,2,\ldots\}$. Note that $Z(t+1) = \max\{0, Z(t) + \text{tr}(\mathbf{Q}(t)) - \bar{P}\}$ implies that 
\begin{align*}
Z^2(t+1) \leq &(Z(t) + \text{tr}(\mathbf{Q}(t)) - \bar{P})^2 \\
\leq& Z^2(t) + 2Z(t)\big(\text{tr}(\mathbf{Q}(t)) - \bar{P}\big) + (\text{tr}(\mathbf{Q}(t)) - \bar{P})^2\\
\overset{(a)}{\leq} &Z^2(t) + 2Z(t)\big(\text{tr}(\mathbf{Q}(t) - \bar{P}\big) + \max\{\bar{P}^2, (P-\bar{P})^2\}
\end{align*}
where (a) follows from $\vert \text{tr}(\mathbf{Q}(t)) - \bar{P}\vert \leq  \max\{\bar{P}, P-\bar{P}\}$, which further follows from $0\leq \text{tr}(\mathbf{Q}(t)) \leq P$. Rearranging terms and dividing by factor $2$ yields the desired result.
\end{IEEEproof}

Now, we are ready to present the main proof of Lemma \ref{lm:with-csit-dpp-bound}. Adding $V\log\det(\mathbf{I} + \mathbf{H}(t)\mathbf{Q}(t)\mathbf{H}\herm(t))$ to both sides in \eqref{eq:drift-bound} yields
\begin{align*}
&-\Delta(t) + V\log\det(\mathbf{I} + \mathbf{H}(t)\mathbf{Q}(t)\mathbf{H}\herm(t)) \\
\geq& V\log\det(\mathbf{I} + \mathbf{H}(t)\mathbf{Q}(t)\mathbf{H}\herm(t)) -Z(t) \big(\text{tr}(\mathbf{Q}(t)) - \bar{P}\big) \\ &- \frac{1}{2}\max\{\bar{P}^2, (P-\bar{P})^2\}\\
\overset{(a)}{\geq} &V\log\det(\mathbf{I} + \mathbf{H}(t)\mathbf{Q}^\ast(t)\mathbf{H}\herm(t)) - Z(t) \text{tr}(\mathbf{Q}^\ast(t) - \bar{P}) \\&- \frac{1}{2}\max\{\bar{P}^2, (P-\bar{P})^2\} - 2VP\sqrt{N_T} (2B+\delta) \delta
\end{align*}
where (a) follows from Lemma \ref{lm:with-csit-relate-omega-only}.
 
Taking expectations on both sides yields
\begin{align*}
&- \mathbb{E}[\Delta(t)] +V \mathbb{E}[R(t)] \\
\geq& V R^{\text{opt}} - \mathbb{E}[Z(t) (\text{tr}(\mathbf{Q}^\ast(t)) - \bar{P})] - \frac{1}{2}\max\{\bar{P}^2, (P-\bar{P})^2\} \\ &- 2VP\sqrt{N_T} (2B+\delta)\delta\\
\overset{(a)}{=}&  V R^{\text{opt}} - \mathbb{E} [\mathbb{E}[Z(t) (\text{tr}(\mathbf{Q}^\ast(t)) - \bar{P}) |Z(t)]] \\&- \frac{1}{2}\max\{\bar{P}^2, (P-\bar{P})^2\} - 2VP\sqrt{N_T} (2B+\delta)\delta\\
\overset{(b)}{\geq}& V R^{\text{opt}} - \frac{1}{2}\max\{\bar{P}^2, (P-\bar{P})^2\} - 2VP\sqrt{N_T} (2B+\delta)\delta
\end{align*}
where (a) follows by noting that $\mathbb{E}[Z(t) (\text{tr}(\mathbf{Q}^\ast(t)) - \bar{P}) |Z(t)]$ is the expectation conditional on $Z(t)$ and the iterated law of expectations; and (b) follows from $\mathbb{E}[Z(t) \text{tr}(\mathbf{Q}^\ast(t) - \bar{P}) |Z(t)] = Z(t)\mathbb{E}[\text{tr}(\mathbf{Q}^\ast(t)) - \bar{P}] \leq 0$, where the identity follows because $\mathbf{Q}^\ast(t)$ only depends on $\mathbf{H}(t)$ and is independent of $Z(t)$, and the inequality follows because $Z(t)\geq 0$ and $\mathbb{E}[\text{tr}(\mathbf{Q}^\ast(t)) - \bar{P}]\leq 0,\forall t$.

Rearranging terms and dividing both sides by $V$ yields $-\frac{1}{V}\mathbb{E}[\Delta(t)] +\mathbb{E}[R(t)] \geq R^{\text{opt}} - \frac{\max\{\bar{P}^2, (P-\bar{P})^2\}}{2V} - 2P\sqrt{N_T} (2B+\delta)\delta$.

\section{Proof of Lemma \ref{lm:no-csit-transmit-covaraince-update}} \label{app:pf-lm-no-csit-transmit-covaraince-update}
A problem similar to problem \eqref{eq:no-csit-primal-opt-obj}-\eqref{eq:no-csit-primal-opt-sdp} (with inequality constraint \eqref{eq:no-csit-primal-opt-trace} replaced by the equality constraint $\text{tr}(\mathbf{Q}) = \bar{P}$) is considered in Lemma 14 in \cite{Palomar09TSP}.  The problem in \cite{Palomar09TSP} is different from \eqref{eq:no-csit-primal-opt-obj}-\eqref{eq:no-csit-primal-opt-sdp} since inequality constraint \eqref{eq:no-csit-primal-opt-trace} is not necessarily tight at the optimal solution to \eqref{eq:no-csit-primal-opt-obj}-\eqref{eq:no-csit-primal-opt-sdp}. However, the proof flow of the current lemma is similar to \cite{Palomar09TSP}. We shall first reduce problem \eqref{eq:no-csit-primal-opt-obj}-\eqref{eq:no-csit-primal-opt-sdp} to a simpler convex program with a real vector variable by characterizing the structure of its optimal solution. After that, we can derive an (almost) closed-form solution to the simpler convex program by studying its KKT conditions. The details of the proof are as follows:

\begin{Claim}
If $\widehat{\boldsymbol{\Theta}}$ is an optimal solution to the following convex program:
\begin{align}
\min \quad & \frac{1}{2} \Vert \boldsymbol{\Theta} - \boldsymbol{\Sigma} \Vert^2_F \label{eq:app-no-csit-primal-opt1-obj}\\
\text{s.t.} \quad  & \text{tr}(\boldsymbol{\Theta}) \leq \bar{P} \label{eq:app-no-csit-primal-opt1-trace}\\
			 & \boldsymbol{\Theta} \in \mathbb{S}^{N_{T}}_+ \label{eq:app-no-csit-primal-opt1-sdp}
\end{align}
then $\widehat{\mathbf{Q}} = \mathbf{U}\herm\widehat{\boldsymbol{\Theta}}\mathbf{U}$ is an optimal solution to problem \eqref{eq:no-csit-primal-opt-obj}-\eqref{eq:no-csit-primal-opt-sdp}.
\end{Claim}
\begin{IEEEproof}
This claim can be proven by contradiction. Let $\widehat{\boldsymbol{\Theta}}$ be an optimal solution to convex program \eqref{eq:app-no-csit-primal-opt1-obj}-\eqref{eq:app-no-csit-primal-opt1-sdp} and define $\widehat{\mathbf{Q}} = \mathbf{U}\herm\widehat{\boldsymbol{\Theta}} \mathbf{U}$. Assume that there exists $\widetilde{\mathbf{Q}}\in \mathbb{S}^{N_{T}}_{+}$ such that $\widetilde{\mathbf{Q}}\neq \widehat{\mathbf{Q}}$ and is a solution to problem \eqref{eq:no-csit-primal-opt-obj}-\eqref{eq:no-csit-primal-opt-sdp} that is strictly better than $\widehat{\mathbf{Q}}$. Consider $\widetilde{\boldsymbol{\Theta}} = \mathbf{U} \widetilde{\mathbf{Q}} \mathbf{U}\herm$ and reach a contradiction by showing $\widetilde{\boldsymbol{\Theta}}$ is strictly better than $\widehat{\boldsymbol{\Theta}}$ as follows:

Note that $\text{tr}(\widetilde{\boldsymbol{\Theta}}) = \text{tr}( \mathbf{U} \widetilde{\mathbf{Q}} \mathbf{U}\herm) = \text{tr}(\widetilde{\mathbf{Q}}) \leq \bar{P}$, where the last inequality follows from the assumption that $\widetilde{\mathbf{Q}}$ is solution to problem  \eqref{eq:no-csit-primal-opt-obj}-\eqref{eq:no-csit-primal-opt-sdp}. Also note that $\widetilde{\boldsymbol{\Theta}} \in \mathbb{S}^{N_{T}}_{+}$ since $\widetilde{\mathbf{Q}}\in \mathbb{S}^{N_{T}}_{+}$. Thus, $\widetilde{\boldsymbol{\Theta}}$ is feasible to problem  \eqref{eq:app-no-csit-primal-opt1-obj}-\eqref{eq:app-no-csit-primal-opt1-sdp}.

Note that  $\Vert \widetilde{\boldsymbol{\Theta}} - \boldsymbol{\Sigma}\Vert_F \overset{(a)}{=} \Vert \mathbf{U}\herm \widetilde{\boldsymbol{\Theta}} \mathbf{U} - \mathbf{U}\herm  \boldsymbol{\Sigma} \mathbf{U} \Vert_F 
\overset{(b)}{=} \Vert\widetilde{\mathbf{Q}}  - \mathbf{X}\Vert_F 
\overset{(c)}{<} \Vert\widehat{\mathbf{Q}}  - \mathbf{X}\Vert_F 
\overset{(d)}{=}  \Vert \mathbf{U}\widehat{\mathbf{Q}}\mathbf{U}\herm  - \mathbf{U}\mathbf{X} \mathbf{U}\herm\Vert_F
\overset{(e)}{=}  \Vert\widehat{\boldsymbol{\Theta}} - \boldsymbol{\Sigma} \Vert_{F}$,  where (a) and (d) follow from the fact that Frobenius norm is unitary invariant\footnote{That is $\Vert \mathbf{A} \mathbf{U}\Vert_{F} = \Vert \mathbf{A}\Vert_{F}$ for all $\mathbf{A}\in \mathbb{C}^{n\times n}$ and all unitary matrix $\mathbf{U}$. }; (b) follows from the fact that $\widetilde{\boldsymbol{\Theta}} = \mathbf{U} \widetilde{\mathbf{Q}} \mathbf{U}\herm$ and $\mathbf{X} = \mathbf{U}\herm\boldsymbol{\Sigma}\mathbf{U}$; (c) follows from the fact that $\widetilde{\mathbf{Q}}$ is strictly better than $\widehat{\mathbf{Q}}$; and (e) follows from the fact that $\widehat{\mathbf{Q}} = \mathbf{U}\herm\widehat{\boldsymbol{\Theta}}\mathbf{U}$ and $\mathbf{X} = \mathbf{U}\herm\boldsymbol{\Sigma}\mathbf{U}$. Thus, $\widetilde{\boldsymbol{\Theta}}$ is strictly better than $\widehat{\boldsymbol{\Theta}}$. A contradiction!
\end{IEEEproof}

\begin{Claim}
The optimal solution to problem \eqref{eq:app-no-csit-primal-opt1-obj}-\eqref{eq:app-no-csit-primal-opt1-sdp} must be a diagonal matrix. 
\end{Claim}
\begin{IEEEproof}
This claim can be proven by contradiction. Assume that problem \eqref{eq:app-no-csit-primal-opt1-obj}-\eqref{eq:app-no-csit-primal-opt1-sdp} has an optimal solution $\widetilde{\boldsymbol{\Theta}}$ that is not diagonal.  Since $\widetilde{\boldsymbol{\Theta}}$ is positive semidefinite, all the diagonal entries of $\widetilde{\boldsymbol{\Theta}}$ are non-negative. Define $\widehat{\boldsymbol{\Theta}}$ as a diagonal matrix whose the $i$-th diagonal entry is equal to the $i$-th diagonal entry of $\widetilde{\boldsymbol{\Theta}}$ for all $i\in\{1,2,\ldots,N_{T}\}$. Note that $\text{tr}(\widehat{\boldsymbol{\Theta}}) = \text{tr}(\widetilde{\boldsymbol{\Theta}}) \leq \bar{P}$ and $\widehat{\boldsymbol{\Theta}} \in \mathbb{S}^{n}_+$. Thus, $\widehat{\boldsymbol{\Theta}}$ is feasible to problem \eqref{eq:app-no-csit-primal-opt1-obj}-\eqref{eq:app-no-csit-primal-opt1-sdp}. Note that $\Vert \widehat{\boldsymbol{\Theta}} - \boldsymbol{\Sigma}\Vert_F < \Vert \widetilde{\boldsymbol{\Theta}} - \boldsymbol{\Sigma}\Vert_F $ since $\boldsymbol{\Sigma}$ is diagonal. Thus,  $\widehat{\boldsymbol{\Theta}}$ is a solution strictly better than $\widetilde{\boldsymbol{\Theta}}$. A contradiction!  So the optimal solution to problem \eqref{eq:app-no-csit-primal-opt1-obj}-\eqref{eq:app-no-csit-primal-opt1-sdp} must be a diagonal matrix.
\end{IEEEproof}

By the above two claims, it suffices to assume that the optimal solution to problem \eqref{eq:no-csit-primal-opt-obj}-\eqref{eq:no-csit-primal-opt-sdp} has the structure $\hat{\mathbf{Q}} = \mathbf{U}\herm\boldsymbol{\Theta}\mathbf{U}$, where $\mathbf{\Theta}$ is a diagonal with non-negative entries $\theta_{1}, \ldots, \theta_{N_{T}}$. To solve problem \eqref{eq:no-csit-primal-opt-obj}-\eqref{eq:no-csit-primal-opt-sdp}, it suffices to consider the following convex program.
\begin{align}
\min \quad & \frac{1}{2} \sum_{i=1}^{N_{T}} (\theta_i - \sigma_i)^2  \label{eq:app-no-csit-primal-opt2-obj} \\
\text{s.t.} \quad  & \sum_{i=1}^{N_{T}} \theta_i \leq \bar{P} \label{eq:app-no-csit-primal-opt2-trace} \\
			 & \theta_i \geq 0, \forall i\in\{1,2,\ldots,N_{T}\} \label{eq:app-no-csit-primal-opt2-sdp}
\end{align}

Note that problem \eqref{eq:app-no-csit-primal-opt2-obj}-\eqref{eq:app-no-csit-primal-opt2-sdp} satisfies Slater's condition. So the optimal solution to problem \eqref{eq:app-no-csit-primal-opt2-obj}-\eqref{eq:app-no-csit-primal-opt2-sdp} is characterized by KKT conditions \cite{book_ConvexOptimization}.
Introducing Lagrange multipliers $\mu\in \mathbb{R}_{+}$ for inequality constraint $\sum_{i=1}^{N_{T}} \theta_i \leq \bar{P}$ and $\boldsymbol{\nu} = [\nu_1, \ldots, \nu_{N_{T}}]^T\in \mathbb{R}_+^{N_T}$ for inequality constraints $\theta_i\geq 0, i\in\{1,2,\ldots,n\}$. Let $\boldsymbol{\theta}^\ast = [\theta_{1}^{\ast}, \ldots, \theta_{N_{T}}^{\ast}]^{T}$ and $(\mu^\ast, \boldsymbol{\nu}^\ast)$ be any primal and dual pair with the zero duality gap. By KKT conditions, we have $\theta_i^\ast - \sigma_i + \mu^\ast - \nu_i^\ast = 0,\forall i\in\{1,2,\ldots, N_{T}\}; 
\sum_{i=1}^{N_{T}} \theta_{i}^\ast \leq  \bar{P};
\mu^{\ast} \geq 0;
\mu^{\ast} \big[\sum_{i=1}^{N_{T}} \theta_{i}^\ast -  \bar{P}\big] = 0;
\theta_i^\ast \geq 0, \forall i\in\{1,2,\ldots, N_{T}\};
\nu_i^\ast \geq 0,  \forall i\in\{1,2,\ldots, N_{T}\};
\nu_i^\ast \theta_i^\ast =0,  \forall i\in\{1,2,\ldots, N_{T}\}$.

Eliminating $\nu_i^\ast, \forall i\in\{1,2,\ldots, N_T\}$ in all equations yields
$\mu^\ast \geq \sigma_i - \theta_i^\ast, i\in\{1,2,\ldots, N_{T}\};
\sum_{i=1}^{N_{T}} \theta_{i}^\ast \leq  \bar{P};
\mu^{\ast} \geq 0;
\mu^{\ast} \big[\sum_{i=1}^{N_{T}} \theta_{i}^\ast -  \bar{P}\big] = 0;
\theta_i^\ast \geq 0, \forall i\in\{1,2,\ldots, N_{T}\};
(\theta_i^\ast - \sigma_i + \mu^\ast) \theta_i^\ast =0, \forall i\in\{1,2,\ldots, N_{T}\}$.

For all $ i\in\{1,2,\ldots, N_{T}\}$, we consider $\mu^\ast < \sigma_i$ and $\mu^\ast \geq \sigma_i$ separately:
\begin{enumerate}
\item If $\mu^\ast < \sigma_i$ , then $\mu^\ast \geq \sigma_i - \theta_i^\ast$ holds only when $\theta_i^\ast >0$, which by $(\theta_i^\ast - \sigma_i + \mu^\ast) \theta_i^\ast =0$ implies that $\theta_i^\ast = \sigma_i - \mu^\ast$.
\item If $\mu^\ast \geq \sigma_i$, then $\theta_i^\ast >0$ is impossible, because $\theta_i^\ast >0$ implies that $\theta_i^\ast -\sigma_i +\mu^\ast >0$, which together with $\theta_i^\ast >0$ contradicts the slackness condition $(\theta_i^\ast - \sigma_i + \mu^\ast) \theta_i^\ast =0$. Thus, if $\mu^\ast \geq \sigma_i$, we must have $\theta_i^\ast = 0$.
\end{enumerate}
Summarizing both cases, we have $\theta_i^\ast = \max[0,\sigma_i - \mu^\ast], \forall i\in\{1,2,\ldots,N_{T}\}$, where $\mu^{\ast}$ is chosen such that $\sum_{i=1}^{N_{T}} \theta_{i}^\ast \leq  \bar{P}$, $\mu^{\ast} \geq 0$ and $\mu^{\ast} \big[\sum_{i=1}^{N_{T}} \theta_{i}^\ast -  \bar{P}\big] = 0$.

To find such $\mu^{\ast}$, we first check if $\mu^{\ast} =0$. If $\mu^{\ast} =0$ is true, the slackness condition $\mu^{\ast} \big[\sum_{i=1}^{N_{T}} \theta_{i}^\ast -  \bar{P}\big]$ is guaranteed to hold and we need to further require $\sum_{i=1}^{N_{T}} \theta_{i}^{\ast} = \sum_{i=1}^{N_{T}} \max[0,\sigma_{i}] \leq \bar{P}$. Thus $\mu^{\ast} =0$ if and only if $\sum_{i=1}^{n} \max [0,\sigma_{i}] \leq \bar{P}$. Thus, Algorithm \ref{alg:no-csit-primal-opt-exact-alg} check if $\sum_{i=1}^{N_{T}} \max[0,\sigma_{i}] \leq \bar{P}$ holds at the first step and if this is true, then we conclude $\mu^{\ast} = 0$ and we are done!

Otherwise, we know $\mu^{\ast} > 0$. By the slackness condition $\mu^{\ast} \big[\sum_{i=1}^{N_{T}} \theta_{i}^\ast -  \bar{P}\big]=0$, we must have $\sum_{i=1}^{N_{T}} \theta_{i}^\ast = \sum_{i=1}^{N_{T}} \max[0,\sigma_{i} - \mu^{\ast}] =\bar{P}$. To find $\mu^{\ast}>0$ such that  $\sum_{i=1}^{N_T} \max[0,\sigma_{i} - \mu^{\ast}] =\bar{P}$, we could apply a bisection search by noting that all $\theta_{i}^{\ast}$ are decreasing with respect to $\mu^{\ast}$.

Another algorithm of finding $\mu^{\ast}$ is inspired by the observation that if $\sigma_{j} \geq \sigma_{k}, \forall j,k\in\{1,2,\ldots, N_T\}$, then $\theta_{j}^{\ast} \geq \theta_{k}^{\ast}$. Thus, we first sort all $\sigma_{i}$ in a decreasing order, say $\pi$ is the permutation such that $\sigma_{\pi(1)} \geq \sigma_{\pi(2)} \geq \cdots \geq \sigma_{\pi(N_T)}$; and then sequentially check if $i\in\{1,2,\ldots,N_T\}$ is the index such that $\sigma_{\pi(i)} - \mu^{\ast} \geq 0$ and $\sigma_{\pi(i+1)} - \mu^{\ast} < 0$. To check this, we first assume $i$ is indeed such an index and solve the equation $\sum_{j=1}^{i} \big[ \sigma_{\pi(j)} - \mu^{\ast} \big] = \bar{P}$ to obtain $\mu^{\ast}$; (Note that in Algorithm \ref{alg:no-csit-primal-opt-exact-alg}, to avoid recalculating the partial sum $\sum_{j=1}^{i} \sigma_{\pi(j)}$ for each $i$, we introduce the parameter $S_i =\sum_{j=1}^{i} \sigma_{\pi(j)}$ and update $S_{i}$ incrementally. By doing this, the complexity of each iteration in the loop is only $O(1)$.) then verify the assumption by checking if $\mu^{\ast}\geq 0$, $\sigma_{\pi(i)} - \mu^{\ast} \geq 0$ and $\sigma_{\pi(i+1)} - \mu^{\ast} \leq 0$. The algorithm is described in Algorithm \ref{alg:no-csit-primal-opt-exact-alg} and has complexity $O(N_{T}\log(N_{T}))$. The overall complexity is dominated by the step of sorting all $\sigma_{i}$.

\section{Proof of Lemma \ref{lm:gradient-approx-bound}} \label{app:pf-lm-gradient-approx-bound}

\subsection{Proof of part 1:}
The boundedness of $\mathbf{D}(t-1)$ can be shown as follows. $\Vert \mathbf{D}(t-1)\Vert_F 
=\Vert  \mathbf{H}\herm(t-1) (\mathbf{I}_{N_{R}} + \mathbf{H}(t-1)\mathbf{Q}(t-1) \mathbf{H}\herm(t-1))^{-1}  \mathbf{H}(t-1)\Vert_F
\overset{(a)}{\leq}\Vert \mathbf{H}(t-1)\Vert_F^2\cdot \Vert(\mathbf{I}_{N_{R}} + \mathbf{H}(t-1)\mathbf{Q}(t-1) \mathbf{H}\herm(t-1))^{-1} \Vert_F   
\overset{(b)}{\leq} \sqrt{N_{R}} B^2 $,  where (a) follows from Fact \ref{fact:1} and (b) follows from $\Vert \mathbf{H}(t-1)\Vert_{F}\leq B$ and Fact \ref{fact:identiy-plus-sdp-inverse-frobenius-bound}.

\subsection{Proof of part 2:}

To simplify the notation, this part uses $\mathbf{H}$, $\widetilde{\mathbf{H}}$ and $\mathbf{Q}$ to represent $\mathbf{H}(t-1)$, $\widetilde{\mathbf{H}}(t-1)$ and $\mathbf{Q}(t-1)$, respectively.  
 
Note that 
\begin{align}
&\Vert  \mathbf{D}(t-1) - \widetilde{\mathbf{D}}(t-1) \Vert_{F} \nonumber\\
=& \Vert \mathbf{H}\herm \big(\mathbf{I}_{N_{R}} + \mathbf{H}\mathbf{Q} \mathbf{H}\herm\big)^{-1}\mathbf{H} -  \widetilde{\mathbf{H}}\herm \big(\mathbf{I}_{N_{R}} +  \widetilde{\mathbf{H}}\mathbf{Q} \widetilde{\mathbf{H}}\herm\big)^{-1}\widetilde{\mathbf{H}}\Vert_{F} \nonumber\\
\leq &\Vert \mathbf{H}\herm \big(\mathbf{I}_{N_{R}} + \mathbf{H}\mathbf{Q} \mathbf{H}\herm\big)^{-1}\mathbf{H} - \widetilde{\mathbf{H}}\herm \big(\mathbf{I}_{N_{R}} + \mathbf{H}\mathbf{Q} \mathbf{H}\herm\big)^{-1}\mathbf{H} \Vert_{F} \nonumber\\
& +  \Vert\widetilde{\mathbf{H}}\herm \big(\mathbf{I}_{N_{R}} + \mathbf{H}\mathbf{Q} \mathbf{H}\herm\big)^{-1}\mathbf{H} - \widetilde{\mathbf{H}}\herm \big(\mathbf{I}_{N_{R}} + \mathbf{H}\mathbf{Q} \mathbf{H}\herm\big)^{-1} \widetilde{\mathbf{H}}\Vert_{F} \nonumber \\ &+ \Vert \widetilde{\mathbf{H}}\herm \big(\mathbf{I}_{N_{R}} + \mathbf{H}\mathbf{Q} \mathbf{H}\herm\big)^{-1} \widetilde{\mathbf{H}} - \widetilde{\mathbf{H}}\herm \big(\mathbf{I}_{N_{R}} +  \widetilde{\mathbf{H}}\mathbf{Q} \widetilde{\mathbf{H}}\herm\big)^{-1}\widetilde{\mathbf{H}} \Vert_{F} \nonumber \\
\leq&  \Vert\big(\mathbf{I}_{N_{R}} + \mathbf{H}\mathbf{Q} \mathbf{H}\herm\big)^{-1} \Vert_{F} \Vert \mathbf{H}\Vert_{F}  \Vert \mathbf{H} - \widetilde{\mathbf{H}}\Vert_{F} \nonumber \\&+ \Vert\big(\mathbf{I}_{N_{R}} + \mathbf{H}\mathbf{Q} \mathbf{H}\herm\big)^{-1} \Vert_{F}  \Vert \widetilde{\mathbf{H}}\Vert_{F} \cdot \Vert \mathbf{H} - \widetilde{\mathbf{H}}\Vert_{F}  \nonumber\\
 &+ \Vert \widetilde{\mathbf{H}}\Vert_{F}^{2}  \Vert \big(\mathbf{I}_{N_{R}} + \mathbf{H}\mathbf{Q} \mathbf{H}\herm\big)^{-1} - \big(\mathbf{I}_{N_{R}} +  \widetilde{\mathbf{H}}\mathbf{Q} \widetilde{\mathbf{H}}\herm\big)^{-1}\Vert_{F}  \label{eq:pf-gradient-bound-part2-eq1}
\end{align}
where both inequalities follow from Fact \ref{fact:1}.

Since $\Vert \mathbf{H}\Vert_{F} \leq B$ and $\Vert \widetilde{\mathbf{H}} - \mathbf{H}\Vert \leq \delta$, by Fact \ref{fact:1}, we have $\Vert \widetilde{\mathbf{H}} \Vert_{F} \leq B + \delta$. By Fact \ref{fact:identiy-plus-sdp-inverse-frobenius-bound}, we have $\Vert \big(\mathbf{I}_{N_{R}} + \mathbf{H}\mathbf{Q} \mathbf{H}\herm\big)^{-1} \Vert_{F} \leq \sqrt{N_{R}}$. The following lemma from \cite{Palomar09TSP} will be useful to bound $\Vert \big(\mathbf{I}_{N_{R}} + \mathbf{H}\mathbf{Q} \mathbf{H}\herm\big)^{-1} - \big(\mathbf{I}_{N_{R}} +  \widetilde{\mathbf{H}}\mathbf{Q} \widetilde{\mathbf{H}}\herm\big)^{-1}\Vert_{F}$ from above.

\begin{Lem}[Lemma 6 in \cite{Palomar09TSP}] \label{lm:mean-value-theorem-matrix-functions}
Let $\mathbf{F}: \mathcal{D}\subseteq \mathbb{C}^{m\times n} \rightarrow \mathbb{C}^{p\times q}$ be a complex matrix-valued function defined on a convex set $\mathcal{D}$, assumed to be continuous on $\mathcal{D}$ and differentiable on the interior of $\mathcal{D}$, with Jacobian matrix\footnote{The Jacobian matrix is defined as the matrix $\mathbf{D}_{\mathbf{X}} \mathbf{F}(\mathbf{X})$ such that $d\text{vec}(\mathbf{F}(\mathbf{x})) = \mathbf{D}_{\mathbf{X}} \mathbf{F}(\mathbf{X})\,d\text{vec}(\mathbf{X})$. Note that the size of $\mathbf{D}_{\mathbf{X}} \mathbf{F}(\mathbf{X}) $ is $pq \times mn$.} $\mathbf{D}_{\mathbf{X}} \mathbf{F}(\mathbf{X})$. Then, for any given $\mathbf{X}, \mathbf{Y}\in \mathcal{D}$, there exists some $t\in (0,1)$ such that  $\Vert F(\mathbf{Y}) - F(\mathbf{X}) \Vert_{F} \leq \Vert\mathbf{D}_{\mathbf{X}} \mathbf{F}(t\mathbf{Y} + (1-t)\mathbf{X}) \text{vec}(\mathbf{Y}-\mathbf{X})\Vert_{2} \leq \Vert \mathbf{D}_{\mathbf{X}} \mathbf{F}(t\mathbf{Y} + (1-t)\mathbf{X})\Vert_{2,\text{mat}} \Vert \mathbf{Y} - \mathbf{X}\Vert_{F}$, where $\Vert \mathbf{A}\Vert_{2, \text{mat}} $ denotes the spectral norm of $\mathbf{A}$, i.e., the largest singular value of $\mathbf{A}$.
\end{Lem}

Lemma \ref{lm:mean-value-theorem-matrix-functions} is essentially a mean value theorem for complex matrix valued functions.  The next corollary is the complex matrix version of elementary inequality $\vert\frac{1}{1+x} - \frac{1}{1+y}\vert \leq \vert x-y \vert, \forall x,y\geq 0$ and follows directly from Lemma \ref{lm:mean-value-theorem-matrix-functions}.

\begin{Cor}\label{cor:bound-from-mean-value-theoerm}
Consider $\mathbf{F}: \mathbb{S}_{+}^{n} \rightarrow \mathbb{S}_{+}^{n}$ defined via $\mathbf{F}(\mathbf{X}) = (\mathbf{I}_{n} + \mathbf{X})^{-1}$. Then, $\Vert F(\mathbf{Y}) - F(\mathbf{X})\Vert_{F} \leq n \Vert \mathbf{Y} - \mathbf{X}\Vert_{F}, \forall \mathbf{X}, \mathbf{Y} \in \mathbb{S}_{+}^{n}$.
\end{Cor}
\begin{proof}
By \cite{Hjorungnes07TSP,Palomar09TSP}, $d \mathbf{X}^{-1} = - \mathbf{X}^{-1} (d \mathbf{X}) \mathbf{X}^{-1}$. Thus, $d (\mathbf{I} + \mathbf{X})^{-1} = - (\mathbf{I} + \mathbf{X})^{-1} (d \mathbf{X})  (\mathbf{I} + \mathbf{X})^{-1}$. By identity $\text{vec}(\mathbf{A}\mathbf{B}\mathbf{C}) = (\mathbf{C}^{T}\otimes \mathbf{A}) \text{vec}(\mathbf{B})$, where $\otimes$ denotes the Kronecker product, we have $d \text{vec}(\mathbf{F}(\mathbf{X})) = - \big( ((\mathbf{I} + \mathbf{X})^{-1})^{T} \otimes (\mathbf{I} + \mathbf{X})^{-1}\big)\, d\text{vec}(\mathbf{X})$. Thus, $\mathbf{D}_{\mathbf{X}} \mathbf{F}(\mathbf{X}) = - ((\mathbf{I} + \mathbf{X})^{-1})^{T} \otimes (\mathbf{I} + \mathbf{X})^{-1}$.  Note that for all $\mathbf{X}\in \mathbb{S}^{n}_{+}$, $\Vert - ((\mathbf{I} + \mathbf{X})^{-1})^{T} \otimes (\mathbf{I} + \mathbf{X})^{-1}\Vert_{2,\text{mat}} 
\leq  \Vert ((\mathbf{I} + \mathbf{X})^{-1})^{T} \otimes (\mathbf{I} + \mathbf{X})^{-1}\Vert_{F} 
\overset{(a)}{=}  \Vert ((\mathbf{I} + \mathbf{X})^{-1})^{T}\Vert_{F} \cdot \Vert (\mathbf{I} + \mathbf{X})^{-1} \Vert_{F} 
=\Vert (\mathbf{I} + \mathbf{X})^{-1} \Vert_{F} ^{2}
\overset{(b)}{\leq} n$, where (a) follows from the fact that $\Vert \mathbf{A} \otimes \mathbf{B}\Vert_{F} = \Vert \mathbf{A}\Vert_{F} \cdot \Vert \mathbf{B}\Vert_{F}, \forall \mathbf{A}\in \mathbb{C}^{m\times n}, \mathbf{B}\in \mathbb{C}^{n\times l}$ (see Exercise 28, page 253 in \cite{book_TopicsMatrixAnalysis}); and (b) follows Fact \ref{fact:identiy-plus-sdp-inverse-frobenius-bound}.  Applying Lemma \ref{lm:mean-value-theorem-matrix-functions} yields $\Vert F(\mathbf{Y}) - F(\mathbf{X})\Vert_{F} \leq n \Vert \mathbf{Y} - \mathbf{X}\Vert_{F}, \forall \mathbf{X}, \mathbf{Y} \in \mathbb{S}_{+}^{n}$.
\end{proof}

Applying the above corollary yields 
\begin{align*}
&\Vert \big(\mathbf{I}_{N_{R}} + \mathbf{H}\mathbf{Q} \mathbf{H}\herm\big)^{-1} - \big(\mathbf{I}_{N_{R}} +  \widetilde{\mathbf{H}}\mathbf{Q} \widetilde{\mathbf{H}}\herm\big)^{-1}\Vert_{F} \\
\overset{(a)}{\leq}& N_{R} \Vert \mathbf{H}\mathbf{Q} \mathbf{H}\herm - \widetilde{\mathbf{H}}\mathbf{Q} \widetilde{\mathbf{H}}\herm\Vert_{F}\\ 
= & N_{R} \Vert \mathbf{H}\mathbf{Q} \mathbf{H}\herm -\widetilde{\mathbf{H}}\mathbf{Q} \mathbf{H}\herm + \widetilde{\mathbf{H}}\mathbf{Q} \mathbf{H}\herm - \widetilde{\mathbf{H}}\mathbf{Q} \widetilde{\mathbf{H}}\herm\Vert_{F}\\
\overset{(b)}{\leq}& N_{R} \big( \Vert \mathbf{H}\mathbf{Q} \mathbf{H}\herm -\widetilde{\mathbf{H}}\mathbf{Q} \mathbf{H}\herm\Vert_{F} + \Vert \widetilde{\mathbf{H}}\mathbf{Q} \mathbf{H}\herm - \widetilde{\mathbf{H}}\mathbf{Q} \widetilde{\mathbf{H}}\herm\Vert_{F}\big)\\ 
\overset{(c)}{\leq}&  N_{R} \big( \Vert \mathbf{Q}\Vert_{F} \Vert \mathbf{H}\herm\Vert_{F}  \Vert \mathbf{H} - \widetilde{\mathbf{H}}\Vert_{F} + \Vert \widetilde{\mathbf{H}}\Vert_{F}  \Vert \mathbf{Q}\Vert_{F} \Vert \mathbf{H}\herm- \widetilde{\mathbf{H}}\herm\Vert_{F}  \big)\\ 
\overset{(d)}{\leq}& N_{R} \bar{P} ( 2B + \delta)\delta
\end{align*}
where (a) follows from Corollary \ref{cor:bound-from-mean-value-theoerm}; (b) and (c) follows from Fact \ref{fact:1}; and (d) follows from the fact that $\Vert \mathbf{H}\Vert_{F} \leq B$ and $\Vert \widetilde{\mathbf{H}} - \mathbf{H}\Vert_{F} \leq \delta$, $\Vert \widetilde{\mathbf{H}} \Vert_{F} \leq B + \delta$, and the fact that $\Vert \mathbf{Q}\Vert_{F}\leq \text{tr}(\mathbf{Q}) \leq \bar{P}$, which is implied by Fact \ref{fact:frobenius-trace-inequality} and $\mathbf{Q}\in \widetilde{\mathcal{Q}}$.  

Plugging equations $\Vert \widetilde{\mathbf{H}} \Vert_{F} \leq B + \delta$, $\Vert \big(\mathbf{I}_{N_{R}} + \mathbf{H}\mathbf{Q} \mathbf{H}\herm\big)^{-1} \Vert_{F} \leq \sqrt{N_{R}}$ and $\Vert \big(\mathbf{I}_{N_{R}} + \mathbf{H}\mathbf{Q} \mathbf{H}\herm\big)^{-1} - \big(\mathbf{I}_{N_{R}} +  \widetilde{\mathbf{H}}\mathbf{Q} \widetilde{\mathbf{H}}\herm\big)^{-1}\Vert_{F}\leq N_{R} \bar{P} ( 2B + \delta)\delta$ into equation \eqref{eq:pf-gradient-bound-part2-eq1} yields $\Vert\mathbf{D}(t-1) -  \widetilde{\mathbf{D}}(t-1) \Vert_{F} \leq \sqrt{N_{R}} B \delta +  \sqrt{N_{R}}(B + \delta) \delta + (B + \delta)^{2} N_{R} \bar{P} ( 2B + \delta)\delta  
= \big(\sqrt{N_{R}} B +  \sqrt{N_{R}} (B + \delta)  + (B + \delta)^{2} N_{R} \bar{P} ( 2B + \delta)  \big)\delta$.

\subsection{Proof of part 3:}

This part follows from $\Vert \widetilde{\mathbf{D}}(t-1)\Vert_{F} \leq \Vert \widetilde{\mathbf{D}}(t-1) - \mathbf{D}(t-1)\Vert_{F} + \Vert \mathbf{D}(t-1)\Vert_{F}$.

\bibliographystyle{IEEEtran}
%%% argument is your BibTeX string definitions and bibliography database(s)
\bibliography{IEEEfull,mybibfile}

\end{document}